\documentclass[11pt,reprint,aip]{revtex4-1}

\usepackage{graphicx}
\usepackage{amsmath,amssymb}
\usepackage{stmaryrd}
\usepackage{tensor}
\usepackage{bbm}
\usepackage{braket}
\usepackage{comment} 
\usepackage{amsthm}
\usepackage{color}

\newcommand{\RR}{\mathbb{R}}

\newcommand{\BB}{\mathbb{B}}
\newcommand{\Tr}{\operatorname{Tr}}
\newcommand{\prox}{\operatorname{prox}}
\newcommand{\co}{\operatorname{co}}

\newcommand{\argmin}{\operatornamewithlimits{argmin}}

\newcommand{\rmd}{\mathrm{d}}
\renewcommand{\vec}[1]{\mathbf{#1}}

\newcommand{\dom}{\operatorname{dom}}

\newtheorem{definition}{Definition}
\newtheorem{theorem}{Theorem}

\renewcommand{\restriction}{\mathord{\upharpoonright}}

\begin{document}

\bibliographystyle{unsrt} 

\title{Differentiable but exact formulation
  of density-functional theory}

\author{Simen Kvaal}
\email{simen.kvaal@kjemi.uio.no}
\affiliation{Centre for Theoretical and Computational Chemistry, Department of Chemistry, University of Oslo, P.O. Box 1033 Blindern, N-0315 Oslo, Norway}
\author{Ulf Ekstr\"om}
\affiliation{Centre for Theoretical and Computational Chemistry, Department of Chemistry, University of Oslo, P.O. Box 1033 Blindern, N-0315 Oslo, Norway}
\author{Andrew M. Teale}
\affiliation{School of Chemistry, University of Nottingham, University Park, Nottingham, NG7 2RD, UK}
\affiliation{Centre for Theoretical and Computational Chemistry, Department of Chemistry, University of Oslo, P.O. Box 1033 Blindern, N-0315 Oslo, Norway}
\author{Trygve Helgaker}
\affiliation{Centre for Theoretical and Computational Chemistry, Department of Chemistry, University of Oslo, P.O. Box 1033 Blindern, N-0315 Oslo, Norway}
\begin{abstract}
  The universal density functional $F$ of density-functional theory is a complicated and ill-behaved function of the 
  density---in particular, $F$ is not differentiable, making
  many formal manipulations more complicated. Whilst $F$ has been well
  characterized in terms of convex analysis as forming a conjugate
  pair $(E,F)$ with the ground-state energy $E$ via the
  Hohenberg--Kohn and Lieb variation principles, $F$ is nondifferentiable and 
  subdifferentiable only on a small (but dense) set of its domain. In this article, we apply a tool from convex
  analysis, Moreau--Yosida regularization, to construct, for any
  $\epsilon>0$, pairs of conjugate functionals
  $({}^\epsilon\!E,{}^\epsilon\! F)$ that converge to
  $(E,F)$ pointwise everywhere as $\epsilon\rightarrow 0^+$, and such that ${}^\epsilon\!
  F$ is (Fr\'echet) differentiable. For technical reasons, we
  limit our attention to molecular electronic systems in a finite but 
  large box.
  It is noteworthy that no information is lost in the Moreau--Yosida regularization: 
  the physical ground-state  energy $E(v)$ is exactly
  recoverable from the regularized ground-state energy ${}^\epsilon\!E(v)$ in a simple way. All concepts
  and results pertaining to the original $(E,F)$ pair have
  direct counterparts in results for $({}^\epsilon\! E, {}^\epsilon\!
  F)$. The Moreau--Yosida regularization therefore allows for an exact,
  differentiable formulation of density-functional theory.
  In particular, taking advantage of the differentiability of ${}^\epsilon\!F$, a rigorous formulation of 
  Kohn--Sham theory is presented that does not suffer from the noninteracting representability problem in standard Kohn--Sham theory. 
\end{abstract}
\maketitle
\section{Introduction}

Modern density-functional theory (DFT) was introduced by Hohenberg and Kohn 
in a classic paper~\cite{Hohenberg1964} and is now the workhorse of quantum
chemistry and other fields of quantum physics. 
Subsequently, DFT was put on a mathematically firm ground by Lieb using convex analysis.\cite{Lieb1983} 
The central quantity of DFT is the universal density functional $F(\rho)$, which represents the electronic
energy of the system consistent with a given density $\rho$. Clearly, the success of DFT hinges on the 
modelling of $F$, an extremely complicated function of the electron density. 
It is an interesting observation that, over the last two or three decades, $F$ has been modelled sufficiently accurately to make DFT
the most widely applied method of quantum chemistry, 
in spite of the fact that 
Schuch and Verstraete~\cite{Schuch2009} have shown 
how considerations from the field of computational complexity place
fundamental limits on exact DFT: if $F(\rho)$ could be found
efficiently, all NP hard problems would be solvable in polynomial
time, which is highly unlikely.\cite{Garey1979} 

From a mathematical point of view, DFT is neatly formulated using
convex analysis\cite{Lieb1983}: The universal density functional $F(\rho)$ and the
ground-state energy $E(v)$ are related by a conjugation operation,
with the density $\rho$ and external potential $v$ being elements of a
certain Banach space $X$ and its dual $X^*$, respectively. The
functionals $F$ and $E$ are equivalent in the sense that they contain
the same information---each can be generated exactly from the other.

The universal density functional $F$ is convex and lower semi-continuous but otherwise
highly irregular and ill behaved. Importantly, $F$ is everywhere
discontinuous and not differentiable in any sense that justifies
taking the functional derivative in formal expressions---even for the $v$-representable densities,
as pointed out by Lammert.\cite{Lammert2005}
For example, it is common practice to formally differentiate $F$ with respect to the
density, interpreting the functional derivative ``$-\delta
F(\rho)/\delta\rho(\mathbf r)$'' as a scalar potential at~$\mathbf
r$. However, this derivative, a G\^ateaux derivative, does not
exist. 

Together with the problem of $v$-representability, conventional DFT is
riddled with mathematically unfounded assumptions that are, in fact,
probably false. For example, conventional Kohn--Sham theory assumes,
in addition to differentiability of $F$, that, if $\rho$ is $v$-representable for an interacting $N$-electron
system, then $\rho$ is also $v$-representable for the corresponding noninteracting
system.\cite{KohnSham1965} While providing excellent predictive results with modelled
approximate density functionals, it is, from a mathematical
perspective, unclear why Kohn--Sham DFT works at all.

It is the goal of this article to remedy this situation by
introducing a family of regularized DFTs based on a tool from convex
analysis known as the \emph{Moreau envelope} or \emph{Moreau--Yosida
  regularization}. For $\epsilon>0$, the idea is to 
introduce a regularized energy functional ${}^\epsilon\! E$ related
to the usual ground-state energy $E$ by
\begin{equation}
  {}^\epsilon\! E(v) = E(v) - \frac{1}{2}\epsilon\|v\|_2^2,\label{eq:introreg}
\end{equation}
where $\|\cdot\|$ is the usual $L^2$-norm.
The convex conjugate of $^\epsilon\! E$ is the Moreau envelope $^{\epsilon}\!F$ of $F$,
from which the regularized ground-state energy can be obtained 
by a Hohenberg-Kohn minimization over densities: 
\begin{equation}
{^\epsilon}\!E(v) = \inf_{\rho} \left( {^\epsilon}\!F(\rho) + (v \vert \rho) \right).
\end{equation}
where $( v \vert \rho)= \int \! v(\vec{r})\rho(\vec{r})\mathrm d \vec{r}$. The usual Hohenberg--Kohn variation principle
is recovered as $\epsilon\rightarrow 0^+$. Importantly,
the Moreau envelope $^\epsilon\!F(\rho)$ is
\emph{everywhere differentiable} and converges pointwise from below to
$F(\rho)$ as $\epsilon\rightarrow 0^+$. We use the term ``regularized'' for both 
$^\epsilon\! E$ and $^\epsilon\! F$, although it is $^\epsilon\! F$ 
that, as will be shown below, becomes differentiable through the procedure.

A remark regarding the Banach spaces of densities and potentials is here in order. If $v$ is a Coulomb potential, then the regularization term in
Eq.\;\eqref{eq:introreg} becomes infinite. Moreover, the strongest results
concerning the Moreau--Yosida
regularization are obtained in a reflexive setting. The usual
Banach spaces $X = L^1(\RR^3)\cap L^3(\RR^3)$ and $X^* =
L^{3/2}(\RR^3)+L^\infty(\RR^3)$ for densities and potentials,
respectively,\cite{Lieb1983} are therefore abandoned, and both
replaced with the Hilbert space $L^2(\BB_\ell)$, where $\BB_\ell =
[-\ell/2,\ell/2]^3$ is an arbitrarily large but finite box in $\RR^3$. As is
well known, domain truncation represents a well-behaved approximation: as $\ell$
increases, all
eigenvalues converge to the $\RR^3$-limit. Moreover, the continuous spectrum is approximated by
an increasing number of eigenvalues whose spacing converges to zero.  

We observe that, in the box,
the difference $E(v)-{}^\epsilon\!E(v) = \frac{1}{2}\epsilon \Vert v \Vert_2^2$ is
arbitrarily small and \emph{explicitly known}---it
does not relate to the electronic structure of the system and is easily calculated from $v$. 
Nothing is therefore lost in the transition from $(E,F)$ to
$({}^\epsilon\!E,{}^\epsilon\!F)$. On the contrary, we obtain a
structurally simpler theory that allows taking the derivative of
expressions involving the universal functional. Moreover, the 
differentiability of ${}^\epsilon\! F$ implies $v$-representability
of any $\rho$, for noninteracting as well as interacting systems,
as needed for a rigorous formulation of Kohn--Sham theory.
In this paper, we explore the Moreau envelope as applied to DFT, 
demonstrating how every concept of standard DFT has a counterpart
in the Moreau-regularized formulation of DFT and vice versa.

The remainder of the article is organized as follows:
In Sec.\;\ref{sec:prelim}, we review formal DFT and discuss the
regularity issues of the universal functional within the nonreflexive Banach-space setting 
of Lieb.\cite{Lieb1983} 
In preparation for the Moreau--Yosida regularization, we next reformulate
DFT in a truncated domain, introducing the Hilbert space $L^2(\BB_\ell)$
as density and potential space.

The Moreau--Yosida regularization is
a standard technique of convex analysis, applicable to any convex
function such as the universal density functional. We introduce this regularization
in Sec.\;\ref{sec:my}, reviewing its basic mathematical properties. To establish notation, a review of
convex analysis is given in the Appendix; for  
a good textbook of convex analysis in a Hilbert space, with an in-depth
discussion of the Moreau--Yosida regularization, see
Ref.~\onlinecite{BauschkeAndCombettes}.

Following the introduction of the Moreau--Yosida regularization, 
we apply it to DFT in Sec.\;\ref{sec:myDFT} 
and subsequently to Kohn--Sham theory in Sec.\;\ref{sec:ks}. Finally, Sec.\;\ref{sec:conclusion} contains some concluding remarks.

\section{Preliminaries}
\label{sec:prelim}

\subsection{Formal DFT}
\label{sec:formal-dft}

In DFT, we express the Born--Oppenheimer ground-state problem of an $N$-electron 
system in the external
electrostatic potential $v(\vec{r})$ as a problem referring only to the one-electron density $\rho(\vec{r})$. 
The Born--Oppenheimer $N$-electron molecular Hamiltonian is given by
\begin{equation}\label{eq:ham}
  H_\lambda(v) = \hat{T} + \lambda \hat{W} + \hat{v},
\end{equation}
where $\hat{T}$ and $\hat{W}$ are the kinetic-energy and electron-electron 
repulsion operators, respectively, while $\hat{v}$ is a
multiplicative $N$-electron operator corresponding to the scalar potential
$v(\vec{r})$. The scalar $\lambda$ is introduced 
to distinguish between the interacting ($\lambda=1$) and
noninteracting ($\lambda=0$) systems.

By Levy's constrained-search argument,\cite{Levy1979} 
the (fully interacting) ground-state energy,
\begin{equation}
E(v) = \inf_\Psi \left\langle \Psi \vert H_1(v) \vert \Psi \right\rangle,
\label{eq:Ev}
\end{equation}
can be written in the form of a Hohenberg--Kohn variation principle,
\begin{equation}
  E(v) = \inf_{\rho\in \mathcal{I}_N} \left( F(\rho) + (v \vert \rho)  \right),
  \label{eq:E_of_v_dft}
\end{equation}
where $\mathcal{I}_N$ is the set of $N$-representable densities---that is,
$\rho\in\mathcal{I}_N$ if and only if there exists a normalized
$N$-electron wave function with finite kinetic energy and 
density $\rho$. In Eq.\;\eqref{eq:Ev}, the infimum extends over all
properly symmetrized and normalized $\Psi\in H^1(\RR^{3N})$, the
first-order Sobolev space consisting of those functions in 
$L^2(\RR^{3N})$ that have first-order derivatives also in
$L^2(\RR^{3N})$ and therefore have a finite kinetic energy.

Different universal density functionals $F$ can be used 
in~Eq.\;(\ref{eq:E_of_v_dft}),
the only requirement of an admissible functional being that the
correct ground-state energy $E(v)$ is recovered. 

Given that $\int \! \rho (\mathbf r)\mathrm d \mathbf r= N$, it follows that $\mathcal{I}_N\subset
L^1(\RR^3)$. As demonstrated by Lieb in Ref.~\onlinecite{Lieb1983}, the universal density functional $F$ 
can be chosen as a unique
lower semi-continuous convex function with respect to the $L^1(\RR^3)$
topology. (By definition, therefore, $F(\rho) = +\infty$ for any
$\rho\notin\mathcal{I}_N$; see Appendix~A for remarks on extended-valued functions.)
Moreover, by a Sobolev inequality,\cite{Lieb1983} we may embed the
$N$-representable densities in the Banach space $X = L^1(\RR^3) \cap
L^3(\RR^3)$, with norm $\|\cdot\|_X = \|\cdot\|_{L^1} +
\|\cdot\|_{L^3}$ and topological dual $X^* = L^\infty(\RR^3) +
L^{3/2}(\RR^3)$. 
Given that this Banach space $X$ has a stronger topology than $L^1(\mathbb R^3)$, 
a convergent sequence in $X$ converges also in $L^1$. 
From the lower semi-continuity of $F$ in $L^1(\mathbb R^3)$, we then obtain
\begin{align}
  \|\rho_n-\rho\|_X\rightarrow 0 &\Rightarrow
  \|\rho_n-\rho\|_1\rightarrow 0 \nonumber \\ &\Rightarrow
  \liminf_n F(\rho_n) \geq F(\rho),\label{eq:lsc-stronger-top}
\end{align}
implying that $F$ is lower semi-continuous also in the topology of $X$.
We note that the choice $X = L^1 \cap L^3$ is not unique, but it has the virtue that all Coulomb potentials are contained in $X^*$.

On the chosen Banach spaces, the (concave and continuous) 
ground-state energy $E: X^\ast \to \RR\cup\{-\infty\}$ and the (convex and lower semi-continuous) 
universal density functional $F: X \to \RR\cup\{+\infty\}$ are related by the variation principles 
\begin{subequations}
  \begin{alignat}{2}
    E(v) &= \inf_{\rho\in X} \left( F(\rho) + ( v \vert \rho) \right), &\quad v &\in X^* \label{eq:E_of_v2}, \\
    F(\rho) &= \sup_{v\in X^*} \left( E(v) - ( v \vert \rho) \right), &\quad \rho&\in X.
    \label{eq:F_of_rho}
  \end{alignat}
\end{subequations}
In the terminology  of convex analysis (see Appendix A), $\rho\mapsto
F(\rho)$ and $v\mapsto -E(-v)$ are each other's convex Fenchel conjugates.
To reflect the nonsymmetric relationship between $E$ and $F$ 
in Eqs.\;\eqref{eq:E_of_v2} and~\eqref{eq:F_of_rho}, we introduce the
nonstandard but useful mnemonic notation
\begin{subequations}
\begin{align}
F &= E^\vee,  \\
E &= F^\wedge,
\end{align}
\end{subequations}
which is suggestive of the ``shape'' of the {resulting} functions: 
$F^\wedge = E$ is concave, whereas $E^\vee = F$ is convex. 

The density functional $F$ in~Eq.\;\eqref{eq:F_of_rho} is an extension of the universal functional
$F_\text{HK}$ derived by Hohenberg and Kohn,\cite{Hohenberg1964} 
the latter functional having from our perspective the problem that it
is defined only for ground-state densities 
($v$-representable densities) in $\mathcal A_N$, an implicitly defined set that we do
not know how to characterize. 

It can be shown that the functional $F$ defined by Eq.\;\eqref{eq:F_of_rho} is
identical to the constrained-search functional~\cite{Lieb1983}  
$F(\rho) = \inf_{\Gamma
  \mapsto \rho} \Tr(\hat{T}+\lambda\hat{W})\Gamma$, where the
minimization is over all ensemble density matrices $\Gamma$ corresponding to a density $\rho$, constructed from
$N$-electron wave functions with a finite kinetic energy.
A related functional is the (nonconvex) Levy--Lieb constrained
search functional,\cite{Levy1979} 
$F_\text{LL}(\rho) = \inf_{\Psi \mapsto\rho}
\langle \Psi \vert (\hat{T}+\lambda\hat{W}) \vert \Psi \rangle $,
obtained by minimizing over pure states only.  
In any case, Eq.\;\eqref{eq:F_of_rho} defines the \emph{unique} lower
semi-continuous, convex universal functional such that $F =
(F^\wedge)^\vee$. In fact, any $\bar F$ that satisfies the
condition $({\bar F}^\wedge)^\vee = F$ is an admissible density
functional. In particular, $F_\text{LL}$ and $F_\text{HK}$ are
both admissible, satisfying this requirement when extended from their
domains ($\mathcal{I}_N$ and $\mathcal{A}_N$, respectively) to all of $X$ by
setting them equal to $+\infty$ elsewhere. 

\subsection{Nondifferentiability of $F$}
\label{sec:nonreg}

The Hohenberg--Kohn variation principle in Eq.\;\eqref{eq:E_of_v_dft} is appealing, 
reducing the $N$-electron problem to a problem referring only to one-electron
densities. However, as discussed in the introduction, $F$ is a complicated function. 
In particular, here we consider its nondifferentiability. 

The G\^ateaux derivative is closely related to the
notion of directional derivatives, see Appendix~A. A function $F$ is G\^ateaux
differentiable at $\rho\in X$ if the directional derivative $F'(\rho;\sigma)$ is
linear and continuous in all directions $\sigma\in X$, meaning that there
exists a $\delta
F(\rho)/\delta \rho \in X^*$ such that
\begin{equation}
  F'(\rho;\sigma) := \frac{\mathrm d F(\rho + s\sigma)}{\mathrm ds}
\Big|_{s=0^+} =
  \left(\left.\frac{\delta F(\rho)}{\delta \rho} \right\vert \sigma\right).
\end{equation}
However, $F$ is finite only on 
$\mathcal{I}_N$. In a direction $\sigma \in X$ such that $\int\! (\rho(\mathbf r) + \sigma(\mathbf r))\,\mathrm d \mathbf r \neq
N$, $F(\rho+s \sigma)=+\infty$ for all $s>0$, implying that 
$F'(\rho;\sigma)=+\infty$ and hence that $F$ is not continuous in the direction of $\sigma$. The same argument shows that $F$ is discontinuous 
also in directions $\sigma$ such that the density $\rho + s\sigma$ is negative in a volume of nonzero measure for all $s>0$.

Abandoning strict G\^ateaux differentiability for this reason, we may
at the next step investigate whether the directional derivative exists and is linear for
directions that stay inside $X_N^+$, the subset of $X$ containing all
nonnegative functions that integrate to $N$ electrons. After all, the
discontinuity of $F$ in directions that change the particle number is
typically dealt with using a Lagrange multiplier for
the particle number constraint.
However, Lammert has demonstrated that, even within $X_N^+$, there 
are, for each $\rho$,  directions such that $F'(\rho;\sigma)=+\infty$, 
associated with short-scale but very rapid spatial oscillations in the density (and an infinite kinetic energy).\cite{Lammert2005}

\subsection{Subdifferentiability of $F$}

Apart from lower (upper) semi-continuity of a convex (concave) function,
the minimal useful regularity is not
G\^ateaux differentiability but subdifferentiability (superdifferentiability), see
Appendix~A. Let $f:X\rightarrow \RR\cup\{+\infty\}$ be convex lower semi-continuous. The
subdifferential of $f$ at $x$, $\partial f(x) \subset X^\ast$, is by
definition the collection of slopes of supporting continuous tangent functionals
of $f$ at $x$, known as the \emph{subgradients} of $f$ at $x$, see Fig.~\ref{fig:tangents} in Appendix~A.  
If the graph of $f$ has a ``kink'' at $x$, then there exists 
more than one such subgradient. At a given point $x \in \dom(f)$, the subdifferential $\partial f(x)$ may be empty.
We denote by $\dom(\partial f)$ the set of points $x \in \dom(f)$
such that $\partial f(x)\neq \emptyset$. It is a fact that $\dom(\partial f)$ is dense in $\dom(f)$
when $f$ is a proper lower semi-continuous convex function. 
The superdifferential of a concave function is similarly defined. 

Together with convexity, subdifferentiability is sufficient to characterize
minima of convex functions: A convex lower semi-continuous functional $f:X\rightarrow
\RR\cup\{+\infty\}$ has a global minimum at $x\in X$ if and only if
$0\in \partial f(x)$. Similarly $x\mapsto f(x) + \braket{\varphi,x}$ has a
minimum if and only if $-\varphi\in\partial f(x)$.

Subdifferentiability is a substantially weaker concept than that of
G\^ateaux (or directional) differentiability. Clearly, if $f(x)$ is
G\^ateaux differentiable at $x$, then $\partial
f(x) = \{\delta f(x)/\delta x\}$. However, the converse is not true: in
infinite-dimensional spaces, it is possible that $\partial
f(x)=\{y\}$, a singleton, while $f(x)$ is not differentiable at
$x$. This is so because $\partial f(x)$ being a singleton is not enough
to guarantee continuity of $f$.

In DFT, subdifferentiability has an important
interpretation. Suppose $\rho$ is an ensemble ground-state density of $v$,
meaning that, for all $\rho' \in \mathcal I_N$, we have the inequality
\begin{equation}
  E(v) = F(\rho) + (v \vert \rho) \leq F(\rho') + ( v \vert \rho').
\end{equation}
Then, the subdifferential of $F(\rho)$ at $\rho$ is 
\begin{equation}
  \partial F(\rho) = \{ -v + \mu \;:\; \mu\in \RR \},
\end{equation}
which is a restatement of the first Hohenberg--Kohn theorem:
the potential for which $\rho$ is a ground-state density is unique up
to a constant shift. On the other hand, if $\rho$ is not a
ground-state density for any $v\in X^*$, then $\partial F(\rho) =
\emptyset$. Thus, a nonempty subdifferential is equivalent to
(ensemble) $v$-representability: $\rho \in \dom(\partial F)$ if and only if
$\rho$ is $v$-representable. Denoting the set of ensemble 
$v$-representable densities by $\mathcal B_N$,
we obtain
\begin{equation}
\rho \in \mathcal B_N \; \Longleftrightarrow \;  \partial F(\rho) \neq \emptyset.
\end{equation}
We note that $\mathcal B_N$ is dense in  
$X_N^+$, the subset of $X$ containing all nonnegative functions that integrate to $N$ electrons.

However, even though subdifferentiability is sufficient for many purposes, 
differentiability of $F$ would make formal manipulations easier. 
Moreover, the characterization of $v$-representable
$\rho\in\mathcal B_N$ is unknown and probably dependent on the
interaction strength $\lambda$. These observations motivate the search
for a differentiable regularization of the universal functional.

\subsection{Superdifferentiability of $E$}

Let us briefly consider the superdifferential of $E$, a concave 
continuous (and hence upper semi-continuous) 
function over $X^*$. 
A fundamental theorem of convex analysis states that
\begin{equation}
  -v\in \partial F(\rho) \;\Longleftrightarrow \; \rho \in \partial E(v),
\label{eq:reciprocal}
\end{equation}
where we use 
the same notation for sub- and superdifferentials.
Thus, the potential $v$ has a ground state with density $\rho$ if and
only if $\rho \in \partial E(v)$; if $v$ does not support a ground state, then $\partial E(v)$ is empty. 
Denoting the set of potentials in $X^*$ that support a ground state by $\mathcal V_N$, we obtain:
\begin{equation}
v \in \mathcal V_N \; \Longleftrightarrow \;  \partial E(v) \neq \emptyset.
\end{equation}
If a ground state is nondegenerate, then 
$\partial E(v) = \{ \rho \}$ is a singleton; together with the fact
that $E$ is continuous, it then follows that $E$ is G\^ateaux
differentiable at $v$. On the other hand, if the ground state is degenerate, then
the subdifferential is the convex hull of $g$ ground-state densities:
 \begin{equation}
  \partial E(v) = \co\{ \rho_1, \rho_2,\cdots,\rho_g \},
\end{equation}
and $E$ is not differentiable at this $v$ unless all the $\rho_i$ are
equal---that is, if the degenerate ground states have the same density. For example, in the absence of 
a magnetic field,
the hydrogen atom has the degenerate ground states $1s\alpha$ and $1s\beta$,
with the same density. 

\section{Domain truncation}
\label{sec:domtrunc}

In Sec.~\ref{sec:my}, we outline the mathematical background for
the Moreau--Yosida regularization. Many useful results, such as
differentiability of the Moreau envelope $^\epsilon\!F(\rho)$, are
only available when the underlying vector space $X$ is reflexive or,
even better, when $X$ is a Hilbert space. However, the Banach space $X = L^1(\RR^3)\cap
L^3(\RR^3)$ used in Lieb's formulation of DFT is nonreflexive. In this
section, we truncate the full space $\RR^3$ to a box
$\BB_\ell = [-\ell/2,\ell/2]^3$ of finite volume $\ell^3$, so large 
that the ground state energy of every system of interest
is sufficiently close to the $\RR^3$ limit. What is lost
from this truncation is well compensated for by the fact that we 
may now formulate DFT using the Hilbert space
\begin{equation}
\mathcal{H}_\ell := L^2(\BB_\ell) 
\end{equation}
for both potentials and densities, as we shall now demonstrate.

\subsection{The ground-state problem}

For the spatial domain $\BB_\ell$, the $N$-electron ground-state
problem is a variational search for the lowest-energy 
wave function $\Psi\in H^1_0(\BB_\ell^N)$, the
first-order Sobolev space with vanishing values of the boundary of
$\BB_\ell^N$, the $N$-fold Cartesian product of $\BB_\ell$. The search
is carried out only over the subset of $H^1_0(\BB_\ell^N)$ which is
also normalized and  properly symmetrized: for a total spin projection
of $\hbar(N_{\uparrow}-N_{\downarrow})/2$, the corresponding subset of
wavefunctions is antisymmetric in the $N_{\uparrow}$ first and the
$N_\downarrow$ last particle coordinates separately.

Any potential in the full space,
$\tilde{v}\in L^{3/2}(\RR^3) + L^\infty(\RR^3)$, induces a potential $v = {\tilde{v}}\restriction_{\BB_\ell}\in
L^{3/2}(\BB_\ell) + L^\infty(\BB_\ell)$ in the truncated domain. We
remark that $L^{3/2}(\BB_\ell) + L^\infty(\BB_\ell) = L^{3/2}(\BB_L)$,
with equivalent topologies.
Since the domain is bounded, the
Rellich--Kondrakov theorem\cite{Evans1998} states that
$H^1_0(\BB_\ell^N)$ is compactly embedded in $L^2(\BB^N_\ell)$, which in turn implies that the spectrum of
the Hamiltonian $H_\lambda(v)$ in Eq.\;\eqref{eq:ham} is purely
discrete.\cite{Babuska1989} Thus, for any potential $v$ in the box, one or more ground-state wave functions
$\Psi_v\in H^1_0$ exists.

We next observe that, if $\tilde{v}$ is a Coulomb
potential, then the truncated potential $v$ belongs to $L^2(\BB_\ell)$. Moreover, $L^2(\BB_\ell)
\subset L^{3/2}(\BB_\ell)$ since $\BB_\ell$ is bounded.
It is therefore sufficient to consider the ground-state
energy as a function 
\begin{equation}
  E_\ell : \mathcal{H}_\ell \rightarrow \RR.
\end{equation}
Regarding the continuity of $E_\ell$, we note that the proof given in
Ref.~\onlinecite{Lieb1983} for the continuity of $E$ in the
$L^{3/2}(\mathbb R^3) +L^\infty(\mathbb R^3) $ topology is
equally valid for $E_\ell$ in the 
$L^{3/2}(\mathbb B_\ell) +L^\infty(\mathbb B_\ell) $ topology. 
Convergence in $L^2(\BB_\ell)$ implies convergence in
$L^{3/2}(\BB_\ell) + L^\infty(\BB_\ell)$. Therefore, $E_\ell$ is continuous in
the $L^2(\BB_\ell)$ topology.

We remark that, as $\ell\rightarrow \infty$,
$E_\ell(v)$ converges to the exact, full-space ground-state energy
$E(v)$. On the other hand, the associated eigenfunctions converge if and only if
the full-space ground-state energy $E(v)$ is an eigenvalue, with $v=0$ as a counterexample.

\subsection{Densities and the universal density functional}

Invoking the usual ensemble constrained-search procedure, we obtain 
\begin{equation}
  E_\ell(v) = \inf_{\rho\in\mathcal{I}_N(\BB_\ell)} F_\ell(\rho) + (v \vert \rho),
\end{equation}
where $\mathcal{I}_N(\BB_\ell)$ is the set of $N$-representable
densities: $\rho\in \mathcal{I}_N(\BB_\ell)$ if and only if there exists
a properly symmetrized and normalized $\Psi\in H^1_0(\BB^N_\ell)$ such
that $\Psi\mapsto\rho$. It is straightforward to see that
\begin{equation}
  \begin{split}
    \mathcal{I}_N(\BB_\ell) &= \big\{ \rho \in L^1(\BB_\ell) \; : \; \rho \geq 0 \;\text{(a.e.)},\; 
\\ &\qquad\quad \sqrt{\rho} \in H^1_0(\BB_\ell), \; \textstyle \int\!\!\rho(\mathbf r)\,\mathrm d \mathbf r = N
    \big\}.
  \end{split}
\end{equation}
The density functional $F_\ell$ is completely analogous to the full-space functional
$F$. In particular, $F_\ell$ is lower
semi-continuous in the $L^1(\BB_\ell)$ topology by Theorem~4.4 and
Corollary~4.5 in Ref.~\onlinecite{Lieb1983}.

We remark that $F_\ell(\rho) = F(\rho)$ for any $\rho\in \mathcal{I}_N(\BB_\ell)$, as seen from the fact that,
if $\Psi\in H^1(\RR^{3N})$ and $\Psi\mapsto \rho$ with $\sqrt{\rho}\in
H^1_0(\BB_\ell)$, then we must have $\Psi\in H^1_0(\BB^N_\ell)$.

Since $\BB_\ell$ is bounded, the Cauchy--Schwarz
inequality gives for any measurable $u$,
\begin{equation}\label{eq:cauchy-s}
  \begin{split}
    \|u\|_1 = (1|\; |u| \;) &\leq \|1\|_2 \|u\|_2 =
    |\BB_\ell|^{1/2}\|u\|_2.
  \end{split}
\end{equation}
By an argument similar to that of Eq.\;\eqref{eq:lsc-stronger-top},
$F_\ell$ is now seen to be lower semi-continuous also with
respect to the $L^2(\BB_\ell)$ topology. Note that 
\begin{equation}
\mathcal{I}_N(\BB_\ell) \subset L^1(\mathbb B_\ell) \cap L^3(\mathbb
B_\ell) = L^3(\BB_\ell) \subset L^2(\mathbb B_\ell)
\end{equation}
so that every $N$-representable density is in $L^2(\BB_\ell)$. Since $F_\ell$ is convex and lower semi-continuous on $\mathcal H_\ell = L^2(\mathbb B_\ell)$, 
we may now formulate DFT in the Hilbert space $\mathcal H_\ell$ as 
\begin{subequations}
\begin{align}
  E_\ell(v) &= \inf_{\rho\in \mathcal H_\ell} \left( F_\ell(\rho) + (v \vert \rho) \right), \label{eq:ELvar}\\
  F_\ell(\rho) &= \sup_{v\in \mathcal H_\ell} \left( E_\ell(v) - (v \vert\rho) \right). \label{eq:FLvar}
\end{align}
\end{subequations}
Given that Hilbert spaces possess a richer structure than Banach spaces, this formulation of DFT is particularly convenient:
densities and potentials are now elements of the same vector space $\mathcal H_\ell$ and reflexivity is guaranteed. 

Even for the full space, $\mathcal{I}_N(\mathbb R^3) \subset
L^2(\RR^3)$, indicating that it is possible to avoid the use of the
box. Indeed, we may restrict the ground-state energy to potentials
$v\in L^2(\RR^3) \subset L^{3/2}(\RR^3) + L^\infty(\RR^3)$:
\begin{equation}
\tilde{E}: L^2(\RR^3)\rightarrow \RR, \quad
\tilde{E} = E\restriction_{L^2(\RR^3)}, 
\end{equation}
a concave and continuous map. Invoking the theory of conjugation within this reflexive Hilbert-space setting, 
we have a convex lower semi-continuous universal functional 
\begin{equation}
\tilde{F} : L^2(\RR^3)\rightarrow \RR\cup\{+\infty\}, \quad
\tilde{F} = \tilde{E}^\vee = (\tilde{F}^\wedge)^\vee. 
\end{equation}
However, Coulomb potentials are not contained in $L^2(\mathbb R)$. On the other hand, this theory is sufficient for dealing with all truncated
Coulomb potentials, obtained, for example, from the usual Coulomb potentials by setting them equal to zero outside
the box $\mathbb B_\ell$; it is also sufficient when working with Yukawa rather than Coulomb potentials. 


The optimality conditions for the Hohenberg--Kohn and Lieb variation principles in Eqs.\;\eqref{eq:ELvar} and~\eqref{eq:FLvar} are
\begin{equation}
  -v\in \partial F_\ell(\rho) \;\Longleftrightarrow \; \rho \in \partial E_\ell(v).
\label{eq:reciprocalL}
\end{equation}
Denoting 
the set of densities for which $F_\ell$ is subdifferentiable by $\mathcal B_\ell$ (by analogy with $\mathcal B_N$ in $X$) and
the set of potentials for which $E_\ell$ is superdifferentiable by $\mathcal V_\ell$ (by analogy with $\mathcal V_N$ in $X^\ast$), 
we obtain 
\begin{equation}
\mathcal B_\ell \subsetneq \mathcal H_\ell, \quad \mathcal V_\ell = \mathcal H_\ell
\end{equation}
where $\mathcal{B}_\ell$ is dense in the subset of $\mathcal H_\ell$ containing all non{\-}negative functions that integrate to $N$ electrons.
The differentiability properties of $F_\ell$ are
the same as those of $F$ discussed in Section~\ref{sec:nonreg}. 
To introduce differentiability, a further regularization is necessary.

\section{Moreau--Yosida regularization}
\label{sec:my}

In this section, we present the basic theory of Moreau--Yosida
regularization, introducing infimal convolutions in
Section\;\ref{subsec:inf}, Moreau envelopes in
Section\;\ref{subsec:moreau}, proximal mappings in
Section\;\ref{subsec:prox}, and conjugates of Moreau envelopes in
Section\;\ref{subsec:conmoreau}. The results are given
mostly without proofs; for these proofs, we refer to the book by
Bauschke and~Combettes,\cite{BauschkeAndCombettes} 
whose notation we follow closely.

\subsection{Infimal convolution}
\label{subsec:inf}

In preparation for the Moreau--Yosida regularization, we introduce the 
concept of infimal convolution in this section and discuss its properties on a Hilbert space $\mathcal H$.

\begin{definition}
For $f,g : \mathcal H \to \RR\cup\{+\infty\}$, the \emph{infimal convolution} 
is the function $f \boxempty g: \mathcal H \to \RR\cup\{\pm\infty\}$ given by
\begin{equation}
(f \boxempty g)(x) = \inf_{y \in \mathcal H} \left( f(y) + g(x - y) \right).
\end{equation}
\end{definition}
In the context of convex conjugation, the infimal convolution is analogous to
the standard convolution in the context of the Fourier transform. 
Here are some basic properties of the infimal convolution for
functions that do not take on the value $-\infty$:
\begin{theorem}\label{thm:infimal}
  Let $f,g : \mathcal H \to \RR\cup\{+\infty\}$. Then:
  \begin{enumerate}
  \item
    $f\boxempty g = g \boxempty f$;    
  \item
    $\dom (f\boxempty g) = \dom f + \dom g = \{x+x' \;:\; x\in \dom
    f,\; x'\in\dom g\}$;   
  \item
    $(f\boxempty g)^\wedge = f^\wedge + g^\wedge$;     
  \item
    if $f$ and $g$ are convex, then $f\boxempty g$ is convex.  
  \end{enumerate}
\end{theorem}
\begin{proof}
  See Ref~\onlinecite{BauschkeAndCombettes}, Props.~12.6, ~12.11 and~13.21.
\end{proof}

Henceforth, we restrict our attention to \emph{all lower semi-continuous proper convex 
functions} $f: \mathcal H \to \RR\cup\{+\infty\}$, denoting the set of all such functions by
$\Gamma_0(\mathcal H)$, see Appendix~\ref{matsup}. 
We also need the concepts of coercivity and supercoercivity:
a function $f: \mathcal H\rightarrow \RR\cup\{+\infty\}$ is \emph{coercive} if
$f(x)\rightarrow +\infty$ whenever $\|x\|_{\mathcal{H}}\rightarrow +\infty$ and \emph{supercoercive}
if $f(x)/\Vert x \Vert_{\mathcal{H}} \rightarrow +\infty$ whenever $\|x\|_{\mathcal{H}}\rightarrow +\infty$. 
For example, $F_\ell \in \Gamma_0(\mathcal{H}_\ell)$ is coercive, whereas
$-E_\ell \in \Gamma_0(\mathcal{H}_\ell)$ is not coercive. 

For functions in $\Gamma_0(\mathcal H)$, we have
the following stronger properties of the infimal convolution:
\begin{theorem}\label{thm:infimal2}
  Let $f,g \in \Gamma_0(\mathcal H)$ such that either $g$ is supercoercive or $f$ is bounded from below and $g$ is coercive.  Then
  \begin{enumerate}
  \item $f\boxempty g \in \Gamma_0(\mathcal H)$;
  \item  $(f^\wedge + g^\wedge)^\vee = f\boxempty g$;
  \item\label{item:uniqueness}
    for each $x \in \mathcal H$, there exists $x_\ast \in \mathcal H$ such that
\begin{equation}
     (f \; \boxempty \; g)(x) = f(x_*) + g(x-x_*)
\end{equation}
    where $x_*$ is unique if $g$ is strictly convex.
  \end{enumerate}
\end{theorem}
\begin{proof}
  Point~1 follows from Ref.~\onlinecite{BauschkeAndCombettes},
  Prop.~12.14. Point~2 follows from Theorem~\ref{thm:infimal}
  above. Finally, Point 3~follows from the fact that strictly convex functions have
  unique minima; the existence of a minimum follows from the
  (super)coerciveness of the mapping $y \mapsto \|x-y\|_{\mathcal{H}}^2/2$. 
\end{proof}

\subsection{The Moreau envelope}
\label{subsec:moreau}

In the following, we introduce the Moreau envelope of functions in $\Gamma_0(\mathcal H)$ and review
its properties.

\begin{definition}\label{def:myX}
  For $f \in \Gamma_0(\mathcal H)$ and $\epsilon > 0$, the \emph{Moreau--Yosida regularization} or the \emph{Moreau envelope}
  ${^\epsilon}\!f : \mathcal H \to \RR\cup\{+\infty\}$ is the infimal convolution of $f$ with 
$x \mapsto \frac{1}{2\epsilon} \Vert x \Vert^2_{\mathcal H}$: 
    \begin{align} {^\epsilon}\!f(x) &= \inf_{y\in \mathcal H} \left( f(y) +
      \frac{1}{2\epsilon}\|x-y\|_\mathcal H^2 \right).\label{eq:fmy} 
    \end{align}
\end{definition}
Since $f \in \Gamma_0(\mathcal H)$ and since
$x \mapsto \frac{1}{2\epsilon} \Vert x \Vert^2_{\mathcal H}$ is strictly convex and supercoercive,
it follows from Theorem\;\ref{thm:infimal} that ${^\epsilon}\!f \in \Gamma_0(\mathcal H)$. 
In fact, ${^\epsilon}\!f$ is much more well behaved than a general function in
$\Gamma_0(\mathcal H)$, as the following theorem shows. 
\begin{theorem}\label{thm:my}
The Moreau envelope ${^\epsilon}\!f$ of $f \in \Gamma_0(\mathcal H)$ with $\epsilon> 0$
satisfies the following properties:
  \begin{enumerate}
  \item ${^\epsilon}\!f \in \Gamma_0(\mathcal H)$ with $\dom {^\epsilon}\!f = \mathcal H$;
  \item $\inf f(\mathcal H) \leq {^\epsilon}\!f(x) \leq {^\gamma}\!f(x) \leq f(x)$ for all
$x \in \mathcal H$ and all $0 \leq \gamma \leq \epsilon$;
   \item $\inf {^\epsilon}\!f(\mathcal H) = \inf f(\mathcal H)$;
   \item for all $x\in \mathcal H$, ${^\epsilon}\!f(x) \rightarrow
     f(x)$ from below as $\epsilon\rightarrow 0^+$ (even if $x\notin\dom f$);
  \item ${^\epsilon}\!f$ is continuous;
  \item
    ${^\epsilon}\!f$ is Fr\'echet
    differentiable: for every $x \in \mathcal H$, there exists $\nabla {^\epsilon}\!f(x) \in \mathcal H$
    such that for all $y\in \mathcal H$:
    \begin{equation}
      {^\epsilon}\!f(x + y) = {^\epsilon}\!f(x) + \left(\nabla
        {^\epsilon}\!f(x)|y\right) + o\left(\|y\|_\mathcal H\right);
    \end{equation}
  \item  the  subdifferential of ${^\epsilon}\!f$ at $x$ is given by
    \begin{equation}
      \partial {^\epsilon}\!f(x) = \{ \nabla {^\epsilon}\!f(x)\}. 
    \end{equation}
  \end{enumerate}
\end{theorem}
\begin{proof}
  Point\;1 follows from
  Theorems~\ref{thm:infimal} and~\ref{thm:infimal2}. For Points~2 and 3, see
  Ref.\;\onlinecite{BauschkeAndCombettes}, Prop.~12.9. For Point 4, see Prop.~12.32.
  For Points 5--7, see Props.~12.15, 12.28, and 12.29. 
\end{proof}

In Figure~\ref{fig:moreau}, the Moreau envelope is illustrated for a convex function $f$ on the real axis. 
We observe that the minimum value of
$f(x)$ is preserved by the Moreau envelope $^\epsilon f(x)$ and that the second 
argument $x\mapsto \|x-x'\|_{\mathcal{H}}^2/(2\epsilon)$ to the infimal convolution removes all kinks, giving a curvature
equal to that of this function.

\begin{figure}
    \includegraphics{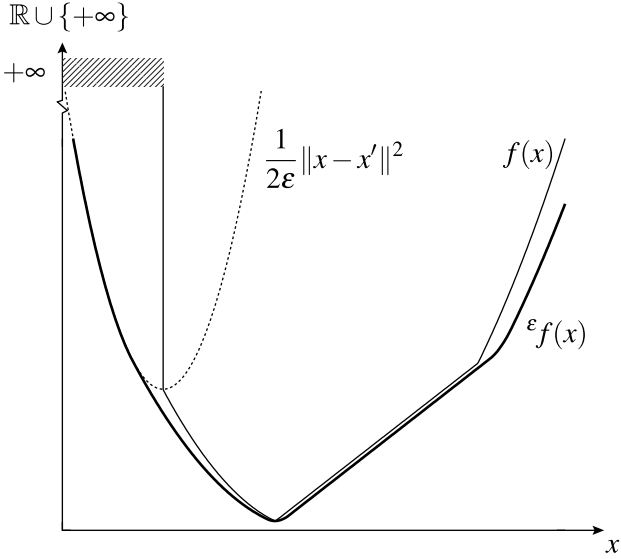}
  \caption{Illustration of the Moreau envelope of a simple convex function
    $f : \RR \rightarrow \RR\cup\{+\infty\}$.  The function $^\epsilon\!f(x)$ is
    plotted in thick lines, whereas $f(x)$ is shown in a
    thinner line. Finally, for a chosen value of $x'$, 
    the function $x\mapsto \|x-x'\|/(2\epsilon)$ is
    superposed on $f(x)$ and ${^\epsilon}\!f(x)$ using a dashed line. \label{fig:moreau}}
\end{figure}

\subsection{The proximal mapping}
\label{subsec:prox}

From Theorem\;\ref{thm:infimal}, it follows that the infimum of ${^\epsilon}\!f(x)$
in Eq.\;\eqref{eq:fmy} is attained with a unique minimizer. We make the following definitions:
\begin{definition}\label{prox} 
Let $f \in \Gamma_0(\mathcal H)$ and $\epsilon > 0$.
The \emph{proximal mapping} $\prox_{\epsilon f} : \mathcal H \to \mathcal H$ is defined by 
\begin{equation}
\prox_{\epsilon f}(x) = \argmin_{y\in \mathcal H} \left( f(y) +
      \frac{1}{2\epsilon}\|x-y\|_\mathcal H^2 \right),\label{eq:prox} 
\end{equation}
where $\prox_{\epsilon f}(x)$ is \emph{the proximal point of $f$ at $x \in \mathcal H$}.
\end{definition}
The usefulness of the proximal mapping follows from the following theorem:
\begin{theorem}\label{thm:my2}
Let $f \in \Gamma_0(\mathcal H)$ and $\epsilon > 0$. Then
  \begin{enumerate}
  \item if $x \in \dom f$ and $\epsilon\rightarrow 0^+$, then
    \begin{equation}
    \left\Vert \prox_{\epsilon f}(x) - x\right\Vert_{\mathcal H}^2 = O(\epsilon);
    \label{eq:proxap}
    \end{equation}
  \item the  Fr\'echet (and G\^ateaux) derivative of ${^\epsilon}\!f$ at $x$ is given by
    \begin{equation}
      \frac{\delta {^\epsilon}\!f(x)}{\delta x} = \nabla {}^\epsilon\!
      f(x) = \epsilon^{-1} \left(x - \prox_{\epsilon f}(x) \right);
    \end{equation}
\item for all $p,x \in \mathcal H$, it holds that
\begin{equation}
p = \prox_{\epsilon f}(x) \; \Longleftrightarrow \;  \epsilon^{-1} (x - p) \in \partial f(p);
\label{eq:proxdif}
\end{equation}
\item if $x \in \mathcal H$, then 
\begin{equation}
   \nabla {^\epsilon}\!f(x) \in \partial f( \prox_{\epsilon f}(x)). 
\end{equation}
\end{enumerate}
\end{theorem}
\begin{proof}
  For Point 1, see the proof of Prop.~12.32 in
  Ref.~\onlinecite{BauschkeAndCombettes}. For Point 2, see
  Prop.~12.29; for Point 3, see Prop.~12.26. Point 4 follows from
  Point 2 and 3.
\end{proof}

\subsection{The conjugate of the Moreau envelope}
\label{subsec:conmoreau}

Given that ${^\epsilon}\!f \in \Gamma_0(\mathcal H)$, there exists a concave ${^\epsilon}\!g \in - \Gamma_0(\mathcal H)$ such that
$({^\epsilon}\!f)^\wedge = {^\epsilon}\!g$ and $({^\epsilon}g)^\vee = {^\epsilon}\!f$. The following theorem gives the basic properties of
this conjugate:
\begin{theorem}\label{thm:cmy}
If ${^\epsilon}\!f$ is the Moreau envelope of $f \in \Gamma_0(\mathcal H)$, then
their conjugates and the superdifferentials of these conjugates are related as 
\begin{subequations}
\begin{align}
({^\epsilon}\!f)^\wedge(x) &= f^\wedge(x) - \frac{1}{2} \epsilon \Vert x \Vert^2_\mathcal H, \label{eq:fw}\\
\partial ({^\epsilon}\!f)^\wedge(x) &= \partial f^\wedge(x) - \epsilon x. \label{eq:fcs}
\end{align}
\end{subequations}
\end{theorem}
\begin{proof}
  Eq.\;\eqref{eq:fw} follows from the fact that the convex conjugate of
  $x\mapsto \|x\|_{\mathcal{H}}^2/(2\epsilon)$ is $x\mapsto \epsilon\|x\|_{\mathcal{H}}^2/2$ and from
  Theorem~\ref{thm:infimal}. Eq.\;\eqref{eq:fcs} follows from the fact
  that the superdifferential of a sum of concave functions 
  is the sum of their superdifferentials if one of the functions is continuous at a common point in their domains, see
  Remark~16.36 of
  Ref.~\onlinecite{BauschkeAndCombettes}. Finally, $\partial (\epsilon\|x\|_{\mathcal{H}}^2/2) = \{\epsilon x \}$.
\end{proof}
Being related in such a simple manner, 
$f^\wedge$ and $({}^\epsilon\!f)^\wedge$ share many properties. We note, however, that $({}^\epsilon\!f)^\wedge$ is strictly concave,
whereas $f^\wedge$ may be merely concave.

We remark that the Moreau envelope is not defined for a \emph{concave}
function $g\in -\Gamma_0(\mathcal{H})$, only for convex
functions. Thus, the notation ${}^\epsilon g$ for a $g\in
-\Gamma_0(\mathcal{H})$ is not to be interpreted as a Moreau envelope,
but as the \emph{concave conjugate} of a Moreau envelope,
${}^\epsilon\!g = ({}^\epsilon(g^\vee))^\wedge$.

\section{Moreau--Yosida regularized DFT}
\label{sec:myDFT}

Having introduced Moreau--Yosida regularization in the preceding section, we
are ready to apply it to DFT on the Hilbert space $\mathcal H_\ell = L^2(\mathbb B_\ell)$.

\subsection{Moreau--Yosida regularized DFT}

Applying Eqs.\;\eqref{eq:fmy} and~\eqref{eq:fw} with $f = F_\ell$ and $f^\wedge = E_\ell$, we obtain
the regularized Lieb functional ${^\epsilon}\!F_\ell: \mathcal H_\ell \to \mathbb R$ and 
ground-state energy ${^\epsilon}\!E_\ell : \mathcal H_\ell \to \mathbb R$,
\begin{subequations}
\begin{align} 
{^\epsilon}\!F_\ell(\rho) &= \inf_{\rho'\in \mathcal H_\ell} \left( F_\ell(\rho') +
      \tfrac{1}{2\epsilon}\|\rho-\rho'\|_2^2 \right), \label{eq:myf} \\
      {^\epsilon}\!E_\ell(v) &= E_\ell(v) - \frac{1}{2}\epsilon\|v\|_2^2. \label{eq:mye}
\end{align}
\end{subequations}
Importantly, these functions are related to each other as conjugate functions; 
just as we have already encountered for the $(E,F)$ and $(E_\ell,F_\ell)$ conjugate pairs.
As such, the following Hohenberg--Kohn and Lieb variation principles hold on the Hilbert space $\mathcal H_\ell$: 
\begin{subequations}
\begin{alignat}{2}
  {^\epsilon}\!E_\ell(v) &= \inf_{\rho\in \mathcal H_\ell} \left( {^\epsilon}\!F_\ell(\rho) + ( v \vert \rho) \right), &\quad \forall v &\in \mathcal H_\ell \label{eq:ER_of_v2}, \\
  {^\epsilon}\!F_\ell(\rho) &= \sup_{v\in \mathcal H_\ell} \left( {^\epsilon}\!E_\ell(v) - ( v \vert \rho) \right), &\quad \forall \rho&\in \mathcal H_\ell. 
  \label{eq:FR_of_rho}
\end{alignat}
\end{subequations}
However, unlike $F$ and $F_\ell$, 
which are finite only for $N$-representable densities, 
the Moreau--Yosida regularized Lieb functional ${^\epsilon}\!F_\ell$ is finite on the whole Hilbert space:
\begin{equation}
\dom({^\epsilon}\!F_\ell) = \mathcal H_\ell
\end{equation}
since, in Eq.\;\eqref{eq:myf}, a finite value is always found on the right-hand side, even when $\rho \notin \mathcal I_N$.
A curious side effect of the regularization is therefore that the minimizing density
in the regularized Hohenberg--Kohn variation principle in Eq.\;\eqref{eq:ER_of_v2} 
(which exists for all $v \in \mathcal H_\ell$) may not
be $N$-representable: it may be negative in a region of
finite measure or contain an incorrect number of electrons.

To illustrate the behaviour of the regularized functional for nonphysical densities, consider $^{\epsilon}\!F_\ell(\rho + c)$
when $\rho$ is $N$-representable and $c \in \mathbb R$. From the definition of the Moreau envelope 
in Eq.\;\eqref{eq:myf}, we obtain straightforwardly that
\begin{equation}
{^\epsilon}\!F_\ell(\rho + c) = {^\epsilon}\!F_\ell(\rho) + \frac{1}{2\epsilon} \ell^3 c^2. \label{eq:Frhoc}
\end{equation}
The regularized density functional thus depends on $c$ in a simple quadratic manner, with a minimum at $c = 0$. As 
$\epsilon$ tends to zero from above, ${^\epsilon}\!F_\ell(\rho + c)$ increases more and more rapidly with increasing $\vert c \vert$, 
approaching $F_\ell(\rho+c)= +\infty$ more closely. As expected, the regularized functional is differentiable
in the direction that changes the number of electrons.

On the face of it, the existence of minimizing `pseudo-densities' in the Hohenberg--Kohn variation principle that are not $N$-representable 
may seem to be a serious shortcoming of the Moreau--Yosida regularization---ideally, we would like the minimizing density to arise from
some $N$-electron wave function. However, the appearance of nonphysical pseudo-densities is an inevitable consequence of the 
regularization---differentiability in all directions 
cannot be achieved without extending the effective domain of $F_\ell$ to all $\mathcal H_\ell$; alternatively,
we may retain the effective domain of $N$-representable densities and instead  work with restricted functional derivatives,
defined only in directions that conserve some properties of the density. Such an approach is straightforward for directions that
change the number of electrons in the system but much more difficult for directions that lead to negative densities or to an
infinite kinetic energy. 

The existence of minimizing pseudo-densities that are not $N$-representable is less important than the fact that 
${^\epsilon}\!F_\ell$ converges pointwise to $F$ from below as $\epsilon \to 0^+$, even when $\rho \notin \mathcal I_N(\BB_\ell)$.
Also, we shall in the next subsection see that
every $\rho \in \mathcal H_\ell$ is linked to a unique physical ground-state density $\rho_\epsilon \in \mathcal B_\ell$. It is therefore possible to
regard (and to treat) the Hohenberg--Kohn minimization over pseudo-densities in $\mathcal H_\ell$ as a minimization over physical densities in $\mathcal B_\ell$,
as discussed below.

We also observe that ${^\epsilon}\!E_\ell$ converges pointwise to $E_\ell$ from below as $\epsilon \rightarrow 0^+$. More importantly, for any
chosen $\epsilon > 0$, we may recover the exact ground-state energy $E_\ell$ from the regularized energy ${^\epsilon}\!E_\ell$ simply
by adding the term $\frac{1}{2}\epsilon\|v\|_2^2$, which \emph{does not depend on the electronic
structure of the system}. Indeed, this term is no more relevant for the molecular electronic system than the
neglected nuclear--nuclear repulsion term---its purpose is merely to make
the ground-state energy strictly concave and supercoercive in the external potential so that the universal density
functional becomes differentiable and continuous. Indeed, no information regarding the electronic system is lost in the 
regularization beyond what is lost upon truncation of the domain from $\mathbb R^3$ to an arbitrarily large cubic box $\mathbb B_\ell$,
needed to make $\frac{1}{2}\epsilon\|v\|_2^2$ finite for all potentials.

\subsection{The proximal density and potential} 

According to the general theory of Moreau--Yosida regularization, a unique minimizer, which we shall here call the \emph{proximal (ground-state) density},
\begin{equation}
\rho_\epsilon = \prox_{\epsilon F_\ell}(\rho) .
\end{equation}
exists for any $\rho \in \mathcal H_\ell$ in the regularized Lieb functional of Eq.\;\eqref{eq:myf}, which may therefore be written as
\begin{equation}
{^\epsilon}\!F_\ell(\rho) = F_\ell(\rho_\epsilon) +
      \frac{1}{2\epsilon}\|\rho - \rho_\epsilon\|_2^2.
\label{eq:proxd} 
\end{equation}
From Eq.\;\eqref{eq:proxdif}, we conclude that the standard  Lieb functional is
subdifferentiable at $\rho_\epsilon$ and hence that $\rho_\epsilon$ is an ensemble $v$-representable ground-state density in $\mathcal H_\ell$:
\begin{equation}
\rho_\epsilon \in \mathcal B_\ell  .
\end{equation}
We also see from
Eq.\;\eqref{eq:proxdif} that every $\rho \in \mathcal H_\ell$ and associated proximal ground-state density $\rho_\epsilon$ together 
satisfy the subgradient relation 
\begin{equation}
\epsilon^{-1} \left(\rho - \rho_\epsilon \right) \in \partial F_\ell(\rho_\epsilon),
\label{eq:veps0}
\end{equation}
implying that 
\begin{equation}
v_\epsilon = \epsilon^{-1} \left(\rho_\epsilon - \rho\right)
\label{eq:veps}
\end{equation}
is an external potential with ground-state density $\rho_\epsilon \in \mathcal B_\ell$.
In the following, we refer to $v_\epsilon$ as the \emph{proximal potential}
associated with $\rho$. We recall that,
by the Hohenberg--Kohn theorem, the density determines the potential up to a constant. The  
subdifferential of $F_\ell$ at the proximal density $\rho_\epsilon$ is therefore 
\begin{equation}
\partial F_\ell(\rho_\epsilon) = - v_\epsilon + \mathbb R .
\end{equation}
where $v_\epsilon$ is the proximal potential of Eq.\;\eqref{eq:veps}. 

Conversely, suppose that $\rho \in \mathcal B_\ell$. There then exists an external potential $v$ such that
$-v \in \partial F_\ell(\rho)$. Expressing $v$ in the form $v = \epsilon^{-1}(\rho - \tilde \rho)$ for some
$\tilde \rho \in \mathcal H_\ell$,
we obtain $\epsilon^{-1}(\tilde \rho - \rho ) \in \partial F_\ell(\rho)$, which by 
Eqs.\;\eqref{eq:proxdif} and~\eqref{eq:veps0} implies that 
$\rho$ is the proximal density of $\tilde \rho$.
Thus, every ensemble $v$-representable density $\rho \in \mathcal B_\ell$ is the proximal
density of $\rho - \epsilon v \in \mathcal H_\ell$ where $v$ is such that $-v \in \partial F_\ell(\rho)$:
\begin{equation}
\rho = \prox_{\epsilon F}(\rho - \epsilon v)  .
\label{eq:rhoisprox}
\end{equation}
In short, we have the important fact that the set of proximal densities in $\mathcal H_\ell$ is
\emph{precisely} the set of ensemble ground-state densities $\mathcal
B_\ell$.
A density $\rho \in \mathcal H_\ell$ whose proximal density is $\rho_\epsilon$ is called a \emph{carrier density} of $\rho_\epsilon$.

By the Hohenberg--Kohn theorem, the potential $v$ in Eq.\;\eqref{eq:rhoisprox} is unique up a constant $c \in \mathbb R$. 
The carrier density is therefore uniquely determined up to an additive constant. The
nonuniqueness of the carrier density also follows directly from Eq.\;(\ref{eq:Frhoc}), which shows that $\rho$ 
and $\rho +c$ where $\rho \in \mathcal H_\ell$ and $c \in \mathbb R$ have the same proximal ground-state density $\rho_\epsilon \in \mathcal B_\ell$.

To summarize, even though the densities in the regularized Hohenberg--Kohn variation principle 
in Eq.\;\eqref{eq:ER_of_v2} are pseudo-densities (not associated with any $N$-electron wave function), 
every such density $\rho \in \mathcal H_\ell$ is uniquely mapped to a ground-state density by the surjective proximal operator 
\begin{equation}
\prox_{\epsilon F}: \mathcal H_\ell \to \mathcal B_\ell.
\end{equation}
This operator performs the decomposition 
\begin{equation}
\rho = \rho_\epsilon - \epsilon v_\epsilon, 
\end{equation}
where the proximal density $\rho_\epsilon \in \mathcal B_\ell$ may be viewed as the `projection' of $\rho$ onto $\mathcal B_\ell$
with potential $v_\epsilon \in \mathcal V_\ell$. We note that $\rho_\epsilon \neq \rho$, even when $\rho \in \mathcal B_\ell$.
The proximal operator is therefore not a true projector.

For any $\rho \in \mathcal H_\ell$, the proximal density $\rho_\epsilon$ and proximal potential $v_\epsilon$ together satisfy
the usual reciprocal relations for the standard  Lieb functional and ground-state energy:
\begin{equation}
-v_\epsilon \in \partial F_\ell(\rho_\epsilon) \; \Longleftrightarrow \; 
\rho_\epsilon \in \partial E_\ell(v_\epsilon),
\end{equation}
see Eq.\;\eqref{eq:reciprocal}, and therefore satisfy the relation:
\begin{equation}
E_\ell(v_\epsilon) = F_\ell(\rho_\epsilon) + (v_\epsilon \vert \rho_\epsilon) .
\label{eq:ELFL}
\end{equation}
Thus, to every solution of the regularized Hohenberg--Kohn variation principle with $-v \in \partial \,{^\epsilon}\!F_\ell(\rho)$ in Eq.\;\eqref{eq:ER_of_v2}
there corresponds a proximal solution to the standard  variation principle with 
$-v_\epsilon \in \partial F_\ell(\rho_\epsilon)$.

\subsection{Differentiability of ${^\epsilon}\!F_\ell$}

Regarding the differentiability of the regularized Lieb functional, 
we note from Theorems\;\ref{thm:my} and~\ref{thm:my2} that ${^\epsilon}\!F_\ell$ is Fr\'echet differentiable so that
\begin{equation}
{^\epsilon}\!F_\ell(\rho + \sigma) = {^\epsilon}\!F_\ell(\rho) -  (v_\epsilon \vert \sigma ) +  o\left(\|\sigma\|_2\right),
\end{equation}
with the derivative given by Eq.\;\eqref{eq:veps}:
\begin{equation}
\nabla \,{^\epsilon}\!F_\ell(\rho) = -v_\epsilon.
\label{eq:FLder}
\end{equation}
G\^ateaux differentiability follows from Fr\'echet 
differentiability: the existence of $\nabla \,{^\epsilon}\!F_\ell(\rho)$ implies that the directional
derivatives at $\rho$ exist in all directions $\sigma\in \mathcal H_\ell$ and are equal to
\begin{equation}
  \frac{\rmd {^\epsilon}\!F_\ell(\rho + t\sigma) }{\rmd t} \Big|_{t = 0} =
  \left( \nabla \, {^\epsilon}\!F_\ell(\rho)  \vert \sigma \right).
\end{equation}
Hence the functional derivative of ${^\epsilon}\!F_\ell$ is well defined and given by
\begin{equation}
\frac{\delta {^\epsilon}\!F_\ell(\rho)}{\delta \rho(\mathbf r)} = - v_\epsilon(\mathbf r),
\label{eq:Fder}
\end{equation}
justifying the formal manipulations involving functional derivatives in DFT,
recalling that ${^\epsilon}\!F_\ell(\rho)$ tends to $F_\ell(\rho)$
pointwise from below as $\epsilon\to 0^+$. (However,
$v_\epsilon$ need not converge to anything.)

\subsection{The optimality conditions of regularized DFT}

The optimality conditions of the regularized DFT variation principles in Eqs.\;\eqref{eq:ER_of_v2} and \eqref{eq:FR_of_rho}
are the reciprocal relations
\begin{equation}
-v \in \partial \,{^\epsilon}\!F_\ell(\rho) \; \Longleftrightarrow \; \rho \in \partial \,{^\epsilon}\!E_\ell(v),
\label{eq:HKLopt} 
\end{equation}
which for the regularized Hohenberg--Kohn variation principle may
now be written in the form of a stationary condition: 
\begin{equation}
\nabla\; {^\epsilon}\!F_\ell(\rho) = - v .
\label{eq:EFstat}
\end{equation}
In combination with Eq.\;\eqref{eq:Fder}, we obtain $v_\epsilon = v$ and hence 
from~Eq.\;\eqref{eq:veps} the following Hohenberg--Kohn stationary condition:
\begin{equation}
\rho = \rho_\epsilon - \epsilon v ,
\end{equation}
suggestive of an iterative scheme with the repeated calculation of the proximal density until self-consistency.

By contrast, the Lieb optimality condition $\rho \in \partial {^\epsilon}\!E_\ell(v)$ in Eq.\;\eqref{eq:HKLopt} cannot be written 
as a stationary condition since the ground-state energy ${^\epsilon}\!E_\ell$ (just like $E$ and $E_\ell$)
is differentiable only when $v$ has a unique ground-state density. From
Theorem\;\ref{thm:cmy}, we obtain 
\begin{equation}
\partial \,{^\epsilon}\!E_\ell(v) = \partial E_\ell(v) - \epsilon v ,
\end{equation}
which shows that the degeneracy of the ground-state energy is preserved by the Moreau--Yosida regularization. 

For any  $\rho \in \mathcal H_\ell$ in~Eq.\;\eqref{eq:EFstat}, 
an explicit expression for the potential $v_\epsilon$ in terms of the proximal density is given in Eq.\;\eqref{eq:veps},
yielding the regularized ground-state energy
\begin{equation}
{^\epsilon}\!E_\ell(v_\epsilon) = {^\epsilon}\!F_\ell(\rho) + (v_\epsilon \vert \rho) .
\end{equation}
Hence, \emph{for every $\rho \in \mathcal H_\ell$, there exists a potential $v_\epsilon$ for which $\rho$ is the ground-state density}. 
Stated differently, the set of ensemble $v$-representable pseudo-densities ${^\epsilon} \mathcal B_\ell$ is equal to the full Hilbert space:
\begin{equation}
{^\epsilon}\mathcal B_\ell = \mathcal H_\ell .
\label{eq:eBL}
\end{equation}
We recall that the proximal density $\rho_\epsilon$ is the exact (standard ) ground-state energy of $v_\epsilon$, see Eq.\;\eqref{eq:ELFL}.

\section{Regularized Kohn--Sham theory}
\label{sec:ks}

In the present section, we apply Moreau--Yosida regularization to Kohn--Sham theory,
beginning with a discussion of the adiabatic connection. The essential point of the regularized Kohn--Sham theory
is the existence of a common ground-state pseudo-density for the interacting and noninteracting systems,
thereby solving the representability problem of Kohn--Sham theory.

In the present section, we simplify notation by omitting the subscript that indicates the length of the box from all quantities---writing $\mathcal H$, for instance, rather
than $\mathcal H_\ell$ everywhere.

\subsection{Regularized adiabatic connection}
\label{sec:ac}

The presentation of Moreau--Yosida regularized DFT given in Section~\ref{sec:myDFT} was for the fully interacting electronic system, with
an interaction strength $\lambda = 1$ in the Hamiltonian of Eq.\;\eqref{eq:ham}. However, given that nothing in the development of 
the theory depends on the value of $\lambda$, it may be repeated without modification for $\lambda \neq 1$. 
In particular, we note that the set of ground-state pseudo-densities is equal to the whole Hilbert space and 
hence is the same for all interaction strengths, see Eq.\;\eqref{eq:eBL}. Consequently,
\emph{every $\rho \in \mathcal H$
is the ground-state pseudo-density of some $v^\lambda \in \mathcal H$, for each $\lambda$}. 

To setup the adiabatic connection, we select $\rho \in \mathcal H$.
Denoting by ${^\epsilon}\!F^\lambda: \mathcal H \to \mathbb R$ the regularized universal
density functional at interaction strength $\lambda$, we obtain from Eq.\;\eqref{eq:EFstat} the unique external potential
\begin{equation}
v_\epsilon^\lambda = - \nabla  {^\epsilon}\!F^\lambda(\rho),
\label{eq:EFstatl}
\end{equation}
for which the regularized ground-state energy at that interaction strength ${^\epsilon}\!E^\lambda : \mathcal H \to \mathbb  R$ 
is given by
\begin{equation}
{^\epsilon}\!E^\lambda(v_\epsilon^\lambda) = {^\epsilon}\!F^\lambda(\rho) + (v_\epsilon^\lambda \vert \rho) .
\label{eq:vfromF}
\end{equation}
As $\lambda$ changes, the potential $v_\epsilon^\lambda$ can be adjusted to 
setup an adiabatic connection of systems
with the same ground-state pseudo-density $\rho$ at different interaction strengths.

In the Moreau--Yosida regularized adiabatic connection, the pseudo-density $\rho$ has a proximal ground-state density that depends on $\lambda$:
\begin{align}
  \rho_\epsilon^\lambda &= \prox_{\epsilon F^\lambda}(\rho) = \rho + \epsilon v_\epsilon^\lambda,  \label{eq:pseudo-to-phys}
\end{align}
which is the true ground-state density in the potential
$v_\epsilon^\lambda$ at that interaction strength:
\begin{equation}
E(v^\lambda_\epsilon) = F(\rho_\epsilon^\lambda) + (v_\epsilon^\lambda \vert \rho_\epsilon^\lambda) .
\label{eq:ELFLx}
\end{equation}
In short, in the adiabatic connection, the effective potential $v_\epsilon^\lambda$ has the same ground-state pseudo-density $\rho$ but different
ground-state densities $\rho_\epsilon^\lambda = \rho + \epsilon v_\epsilon^\lambda$ for different interaction strengths. In the next subsection,
we shall see how this decomposition makes it possible to calculate the true ground-state energy by (regularized) Kohn--Sham theory 
in a rigorous manner, with no approximations except those introduced by domain truncation. 

\subsection{Regularized Kohn--Sham theory}
\label{sec:ks0}

Consider an $N$-electron system with external potential $v_\text{ext} \in \mathcal H$.
We wish to calculate the ground-state energy and to determine a ground-state density of this system:
\begin{equation} 
\rho \in \partial E^1(v_\text{ext}). \label{eq:Enon}
\end{equation}
This can be achieved by solving the interacting many-body Schr\"odinger equation, in some approximate manner. In Kohn--Sham theory,
we proceed differently, solving instead a noninteracting problem with the same density. 

We begin by transforming~Eq.\;\eqref{eq:Enon} into a regularized many-body energy, noting that the energy and superdifferential of the
exact and regularized ground-state energies are related according to Eqs.\;\eqref{eq:fw} and~\eqref{eq:fcs} as
\begin{align}
E^1(v_\text{ext}) &= {^\epsilon}\!E^1(v_\text{ext}) + \frac{1}{2}\epsilon\|v_\text{ext}\|_2^2, \label{eq:mye1} \\
\partial E^1(v_\text{ext}) &= \partial \,{^\epsilon}\!E^1(v_\text{ext}) + \epsilon v_\text{ext} .
\end{align}
From these relations, it follows that the pseudo-density
\begin{equation}
\rho_\text{c} = \rho - \epsilon v_\text{ext} 
\label{eq:rhoast}
\end{equation}
is a ground-state density of the regularized system:
\begin{equation}
\rho_\text{c}  \in \partial \,{^\epsilon}\!E^1(v_\text{ext}). \label{eq:RegI}
\end{equation}
The subscript `c' indicates that $\rho_\text{c}$ is the carrier density of both the physical ground-state of the system $\rho$ according to Eq.\;\eqref{eq:rhoast}
and the ground-state density of the Kohn--Sham system $\rho_\text{s}$:
\begin{equation}
\rho_\text c = \rho_\text s - \epsilon v_\text s.
\end{equation}
Our task is to determine the carrier density and regularized ground-state energy by solving~Eq.\;\eqref{eq:RegI}.
The solution will subsequently be transformed to yield the physical ground-state density and energy. 

We observe that the carrier density $\rho_\text c$ is obtained from the physical 
density $\rho$ by subtracting $\epsilon v_\text{ext}$ with $\epsilon > 0$, see Eq.~(\ref{eq:rhoast}). In practice, $v_\text{ext} < 0$
since the external potential is the attractive Coulomb potential of the nuclei.
It therefore follows that the pseudo-density is strictly positive: $\rho_\text c > 0$.

Given that $\rho_\text{c} \in \mathcal H$, there exists a Kohn--Sham potential $v_\text{s} \in \mathcal H$ such that $\rho_\text{c}$ is
the ground-state density of a noninteracting system in this potential:
\begin{equation} 
\rho_\text{c} \in \partial \,{^\epsilon}\!E^0(v_\text{s}). 
\label{eq:KSeq}
\end{equation}
To determine the regularized Kohn--Sham potential $v_\text{s}$, we first note that the potentials $v_\text{ext}$ and $v_\text{s}$ 
satisfy the stationary condition in Eq.\;\eqref{eq:EFstatl}:
\begin{align}
v_\text{ext} &= - \nabla {^\epsilon}\!F^1(\rho_\text{c}), \\
v_\text{s} &= - \nabla {^\epsilon}\!F^0(\rho_\text{c}).
\end{align}
To proceed, we next introduce
the regularized Hartree--exchange--correlation energy and potential as
\begin{align}
{^\epsilon}\!E_\text{Hxc}(\rho) &= {^\epsilon}\!F^1(\rho) - {^\epsilon}\!F^0(\rho), \\
{^\epsilon}v_\text{Hxc}(\rho) &= \nabla {^\epsilon}\!E_\text{Hxc}(\rho), 
\end{align}
yielding the following expression for the Kohn--Sham potential as a function of the density:
\begin{equation}
v_\text{s} = v_\text{ext} + {^\epsilon}\!v_\text{Hxc}(\rho_\text{c}) .
\end{equation}
To solve the regularized Kohn-Sham problem in Eq.\;\eqref{eq:KSeq}, we
first note that it is related in a simple manner to the standard
Kohn--Sham problem:
\begin{equation} 
\partial \,{^\epsilon}\!E^0(v_\text{s}) = \partial E^0(v_\text{s}) - \epsilon v_\text{s},
\label{eq:KSeqq}
\end{equation}
we then proceed in an iterative fashion.
From some trial pseudo-density $\rho_0$, we iterate 
\begin{subequations}
\begin{align}
v_{i}    &= v_\text{ext} + {^\epsilon}v_\text{Hxc}(\rho_{i-1}) \label{eq:vrho}, \\
\rho_{i} &\in \partial E^0(v_i) - \epsilon v_i, \label{eq:rhov} 
\end{align}
\end{subequations}
until convergence, beginning with $i=1$ and terminating when self-consistency has been established. 
We emphasize that the regularized Kohn--Sham iterations in Eqs.\;\eqref{eq:vrho} and~\eqref{eq:rhov}
are \emph{identical to the iterations in standard Kohn--Sham theory} except for the use of a 
regularized Hartree--exchange--correlation potential in the construction of the Kohn--Sham matrix 
and the subtraction of $-\epsilon v_i$ from the density generated by diagonalization of the resulting Kohn--Sham  matrix.

Having determined the ground-state carrier density $\rho_\text{c}$ and the corresponding Kohn--Sham potential $v_\text{s}$ 
by iterating Eq.\;\eqref{eq:vrho} and~\eqref{eq:rhov} until self consistency,
we calculate the interacting regularized ground-state energy as
\begin{equation}
\begin{split}
{^\epsilon}\!E^1(v_\text{ext}) 
&= {^\epsilon}\!F^1(\rho_\text{c}) + (v_\text{ext} \vert \rho_\text{c})  \\
&= {^\epsilon}\!F^0(\rho_\text{c}) + {^\epsilon}\!E_\text{Hxc}(\rho_\text{c}) + (v_\text{ext} \vert \rho_\text{c})  \\
&= {^\epsilon}\!E^0(v_\text{s}) + (v_\text{ext} - v_\text{s} \vert \rho_\text{c}) + {^\epsilon}\!E_\text{Hxc}(\rho_\text{c}) 
\end{split}
\end{equation}
from which the physical ground-state energy $E^1(v_\text{ext})$ is recovered by adding $\frac{1}{2} \epsilon \Vert v_\text{ext} \Vert^2_{2}$
according to Eq.\;\eqref{eq:mye1}, while the ground-state density $\rho$ is recovered 
by adding $\epsilon v_\text{ext}$ to the pseudo-density $\rho_\text{c}$ according to~Eq.\;\eqref{eq:rhoast}. 
We note that the pair $(\rho_\text{c},v_\text{s})$ is uniquely determined to the extent that $\rho$ in Eq.\;\eqref{eq:Enon} is unique; for systems with degenerate ground-state
densities, several equivalent pairs $(\rho_\text{c},v_\text{s})$ exist. 

By means of Moreau--Yosida regularization, we have thus setup Kohn--Sham theory in a rigorous manner, where the interacting and noninteracting 
ground-state densities are different (by an amount proportional to $\epsilon$) but related by the same carrier density $\rho_\text{c}$, thereby 
solving the noninteracting representability problem of standard Kohn--Sham theory. Moreover, differentiability of the regularized universal density functional
means that the potentials associated with this pseudo-density at different interaction strengths are well defined as the (negative) derivatives of the density functional. 
In the limit where $\epsilon\to 0^+$, standard Kohn--Sham theory is approached, although the limit itself is not expected to be well behaved.

\section{Conclusion}
\label{sec:conclusion}

The possibility of setting up DFT follows from the mathematical properties of the ground-state energy $E(v)$, which is continuous and concave in the external potential $v$. 
By convex conjugation, it may be exactly represented by the lower semi-continuous and convex universal density functional $F(\rho)$, 
whose properties reflect those of the ground-state energy. Unfortunately, $F(\rho)$ depends on the density $\rho$ in a highly irregular manner, being everywhere discontinuous
and nowhere differentiable. These characteristics of $F$ arise in part because $E$ is concave but not strictly concave and not supercoercive. By modifying $E$ in a
way that introduces strict concavity and supercoercivity without losing information about the electronic system, we obtain an alternative DFT, where the universal
density functional is much more well behaved, being everywhere differentiable (and therefore also continuous). This is achieved by Moreau--Yosida regularization, where
we apply convex conjugation not to $E(v)$ itself but to the strictly concave function $E(v) - \frac{1}{2}\epsilon \Vert v \Vert_2^2$, where $\epsilon > 0$. 
The resulting density functional ${^\epsilon}\!F(\rho)$ is convex and differentiable.
Standard DFT is recovered as $\epsilon\to 0^+$ but this limit need not be taken for the theory to be exact---for any chosen value of $\epsilon$, we can
perform DFT as usual; the exact ground-state energy is recovered as $E(v) = {^\epsilon}\!E(v) + \frac{1}{2}\epsilon \Vert v \Vert_2^2$. 
The only restriction on the exact theory is the truncation of the domain from $\mathbb R^3$ to a box of finite (but arbitrarily large) volume; such
a domain truncation simplifies the Moreau--Yosida formulation of DFT by introducing (reflexive) Hilbert spaces of densities and potentials. 

The densities that occur naturally in regularized DFT are not physical densities since they cannot be generated from an $N$-electron wave function in the usual manner.
Nevertheless, each `pseudo-density' $\rho$ has a clear physical interpretation: it can be uniquely decomposed as $\rho = \rho_\epsilon - \epsilon v_\epsilon$, 
where $\rho_\epsilon$ is a physical ground-state density (the `proximal density') and $v_\epsilon$ the associated potential.

This density decomposition justifies Kohn--Sham theory: a given
pseudo-density $\rho$ is uniquely decomposed as $\rho = \rho_\epsilon^\lambda - \epsilon v_\epsilon^\lambda$,
at each interaction strength $\lambda$. As $\lambda$ changes, the decomposition of $\rho$ changes accordingly. 
For the fully interacting system, $\rho = \rho_1 - \epsilon v_\text{ext}$ where $\rho_1$ is the physical ground-state density and $v_\text{ext}$ the external potential; 
for the noninteracting system, $\rho = \rho_\text s - \epsilon v_\text s$, where $\rho_\text s$ and $v_\text s$ are the Kohn--Sham density and potential,
thereby solving the noninteracting representability problem of Kohn--Sham theory. The working equations of regularized Kohn--Sham theory are essentially
identical to those of standard Kohn--Sham theory.

Here, we have considered standard Moreau--Yosida regularization. 
However, we may also consider a generalized approach, in which the regularizing term $\frac{1}{2} \epsilon || v ||^2_2$ is 
replaced by $\frac{1}{2} \epsilon || A v ||^2_2$, where the operator $A$ is chosen based on some \emph{a priori} knowledge of the desired solution. 
Indeed, some choices of $A$ result in approaches closely related to known regularization techniques, such as the Zhao--Morrisson--Parr approach~\cite{ZMP} to calculate the noninteracting universal density functional and the ``smoothing-norm'' regularization approach of Heaton-Burgess \emph{et. al.}, used both in the context of optimized effective potentials~\cite{YangWu,PRLOEPReg} 
and Lieb optimization methods~\cite{WuYang,JCPBulat,US}. These and related approaches will be discussed in a forthcoming paper. 
We expect such Moreau--Yosida techniques to be of great practical value in the implementation of procedures that attempt to determine either the ground-state energy or the universal density functional 
by direct optimization techniques using their derivatives, bearing in mind that both the derivatives and the objective functions are well defined in the regularized context.

\acknowledgments

This work was supported by the Norwegian Research Council through the
CoE Centre for Theoretical and Computational Chemistry (CTCC) Grant
No.\ 179568/V30 and the Grant No.\ 171185/V30 and through the European
Research Council under the European Union Seventh Framework Program
through the Advanced Grant ABACUS, ERC Grant Agreement No.\ 267683.

A.~M.~T.~is also grateful for support from the Royal Society University Research Fellowship scheme.

\appendix
\section{Mathematical Supplement}\label{matsup}

In this section, we review some important concepts of convex analysis
and the calculus of variations. Suggested reading for convex
analysis are van Tiel's book\cite{VanTiel} and the classic text by
Ekeland and T\'emam.\cite{EkelandAndTemam} The present article relies on
additional information gathered in the book by Bauschke and
Combettes,\cite{BauschkeAndCombettes} which focuses on the Hilbert-space formulation of convex analysis.
For functional analysis, the monograph by Kreyszig\cite{Kreyszig} is recommended.

\subsection{Convex functions}

We are here concerned with extended real-valued 
functions $f : X \rightarrow \RR \cup \{\pm\infty\}$ 
over a Banach or Hilbert space $(X,\|\cdot\|_X)$.
Note  that we define $x \pm \infty = \pm\infty$ for any $x\in \RR$, and
$x\cdot\pm\infty = \pm\infty$ for positive real numbers $x$, but that $+\infty
- \infty$ is not defined.

We recall that $X^*$, the topological dual of $X$, is the set of
continuous linear functionals over $X$: if $\varphi\in X^*$,
then $\varphi$ is a real-valued map, continuous and linear in $x\in X$. 
We denote by $\braket{\varphi,x}$ the value of $\varphi$
at $x$, except in the DFT setting, where
the notation $(\cdot|\cdot)$ is used. 
For simplicity, we assume in this section that $X$ is reflexive so that $X^{**} =
X$. Ultimately, we shall work with Hilbert spaces, which are 
reflexive Banach spaces so that $X^* = X$ by the Riesz
representation theorem of functional analysis.

Let $f : X \rightarrow \RR\cup\{+\infty\}$ be an extended-valued
function. The (effective) domain $\dom f$ is the subset
of $X$ where $f$ is not $+\infty$. The function $f$ is said to be proper if $\dom f
\neq \emptyset$.
The function $f$ is convex if, for all $x$ and $y$ in $X$, and for all
$\lambda \in (0,1)$,
\begin{equation}
  f(\lambda x + (1-\lambda)y) \leq \lambda f(x) + (1-\lambda)f(y).
  \label{eq:convexity}
\end{equation}
Note that this formula also makes sense if, say, $f(x) =
+\infty$. The interpretation of convexity is that a linear interpolation
between two points always lays on or above the graph of $f$. We say
that $f$ is strictly convex if strict inequality holds for $x \neq y$ in
Eq.\;(\ref{eq:convexity}). Moreover,
$f$ is said to be concave if the inequality is reversed 
in Eq.\;(\ref{eq:convexity}) and strict concavity is defined
similarly. 

Perhaps the most important property of a convex $f$ is that any local
minimum is also a global minimum. Moreover, if $f$ is strictly convex,
the global minimizer, if it exists, is unique. Convex optimization problems are
in this sense well behaved.

\subsection{Proper lower semi-continuous convex functions}

The minimal useful regularity of convex functions is not
continuity but lower semi-continuity. In a metric space $X$,
a function $f$ is said to be lower
semi-continuous if, for every sequence $\{x_n\}\subset X$
converging to some $x\in X$, we have
\begin{equation}
  f(x) \leq \liminf_n f(x_n).
\end{equation}
The importance of lower semi-continuity is that it guarantees the
existence of a global minimum if $A = \dom f$ is compact: $\inf_{x\in
  A} f(x) = f(x_\text{min})$ for some $x_\text{min}\in A$. For
concave functions, upper semi-continuity is the corresponding useful
notion; $f$ is upper semi-continuous if $-f$ is lower semi-continuous,
by definition.

We are particularly interested in lower semi-continuous proper convex
functions.  The set $\Gamma(X)$ is defined as consisting of all
functions that can be written in the form
\begin{equation}
  f(x) = \sup_{\alpha\in I}  \{ \braket{\varphi_\alpha, x} - g_\alpha\}
\end{equation}
for some family $\{\varphi_\alpha\}_{\alpha\in I} \subset X^*$ of dual
functions and some $\{ g_\alpha\}_{\alpha\in I} \subset \RR$. The set
$\Gamma(X)$ contains precisely all lower semi-continuous proper convex
functions on $X$ and the functions identically equal to $\pm \infty$.  In other words, $f$
is lower semi-continuous proper convex or identically equal to $\pm \infty$ if
and only if it is the pointwise supremum of a set of {continuous}
affine (``straight-line'') functions over $X$. We denote by
$\Gamma_0(X)$ all proper lower semi-continuous functions on $X$:
$\Gamma_0(X) = \Gamma(X) \setminus \{ x\mapsto-\infty, x\mapsto + \infty \}$.
It is a fact that any $f\in \Gamma(X)$ is also \emph{weakly} lower
semi-continuous.

On the dual space $X^\ast$, we denote by $\Gamma^\ast(X^\ast)$ the set of all functions that can be written in the form 
  \begin{equation}
  g(\varphi) = \sup_{\alpha\in I}  \{ \braket{\varphi, x_\alpha} - g_\alpha\} .
  \end{equation}
These functions  are precisely the weak-$*$ lower semi-continuous proper convex functions on $X^\ast$ and the improper functions $\pm \infty$.
The proper functions are $\Gamma_0^\ast(X^\ast) = \Gamma^\ast(X^\ast)
\setminus \{ \varphi\mapsto -\infty, \varphi \mapsto +\infty \}$. 

\begin{theorem}[Convex conjugates]
There is a one-to-one correspondence between the functions $f \in
\Gamma_0(X)$ and the functions $g \in \Gamma_0^\ast(X^\ast)$ given by
\begin{subequations}
  \begin{align}
    f(x) &= \sup_{\varphi\in X^*} \left( \braket{\varphi,x} - g(\varphi) \right), \\
    g(\varphi) &= \sup_{x\in X} \left( \braket{\varphi,x} - f(x)\right) .
  \end{align}
\end{subequations}
\end{theorem}
The unique function $g$ is said to be the \emph{convex conjugate} of $f$ and is denoted by $g = f^\ast$; likewise, $f = g^\ast$ is the \emph{convex conjugate} of $g$.
A pair of functions $f \in \Gamma(X)$ and $g \in \Gamma^\ast(X^\ast)$ that are each other's convex conjugates are said to be \emph{dual functions}. The
dual functions contain the same information, only coded differently: each property of $f$ is reflected, in some manner, in the properties of $f^\ast$ and vice versa.
We note the relations
\begin{equation}
f = (f^{\ast})^\ast = f^{\ast\ast}, \quad
g = (g^{\ast})^\ast = g^{\ast\ast} 
\label{eq:sbc}
\end{equation}
for functions $f \in \Gamma(X)$ and $g \in \Gamma^\ast(X^\ast)$. In
fact, the conjugation operation is a bijective map between $\Gamma(X)$
and $\Gamma^\ast(X^\ast)$, they contain precisely those functions that
satisfy the biconjugation relations in~Eq.\;\eqref{eq:sbc}. 

Because of sign conventions, we work with functions $f \in \Gamma_0(X)$ and $g \in - \Gamma_0^\ast(X^\ast)$. 
It is then convenient to adapt the notation 
\begin{subequations}
\begin{align}
  f^{\wedge}(\varphi) &= \inf_{x\in X} \left( f(x) + \braket{\varphi,x}\right), \\
  g^{\vee}(x) &= \sup_{\varphi\in X^*} \left( g(\varphi) - \braket{\varphi,x}\right) ,
\end{align}
\end{subequations}
for which $f = (f^\wedge)^\vee$ and  $g = (g^\vee)^\wedge$ hold.
In particular, in DFT as developed by Lieb, the density functional and
ground-state energy
\begin{subequations}
\begin{alignat}{2}
F &\in \Gamma_0(X), &\quad X &= L^1\cap L^3 , \\
E &\in -\Gamma^\ast_0(X^\ast), &\quad X^\ast &= L^\infty + L^{3/2},
\end{alignat}
\end{subequations}
are related as $E= F^\wedge$ and $F= E^\vee$.

\subsection{Subdifferentiation}

A dual function $\varphi\in X^*$ is said to be a \emph{subgradient} 
to $f$ at a point $x$ where $f(x)$ is finite if 
\begin{equation}
  f(y) \geq f(x) + \braket{\varphi,y-x}, \quad \forall y\in X,
\end{equation}
meaning that the affine function $y \mapsto f(x) + \braket{\varphi,y-x}$
is nowhere above the graph of $f$. The subdifferential $\partial
f(x)$ is the set of all subgradients to $f$ at $x$, see
Figure~\ref{fig:tangents}. Note that $\partial f(x)$ may be empty. The function $f$
is said to be subdifferentiable at $x\in X$ if $\partial f(x)\neq
\emptyset$.
A function $f\in\Gamma_0(X)$ has a global minimum at $x\in X$ if and only if
$0\in \partial f(x)$. Similarly $x\mapsto f(x) + \braket{\varphi,x}$ has a
minimum if and only if $-\varphi\in\partial f(x)$. A function $f \in \Gamma(X)$ is
subdifferentiable on a dense subset of its domain $\dom(f)$.

\begin{figure}
    \includegraphics{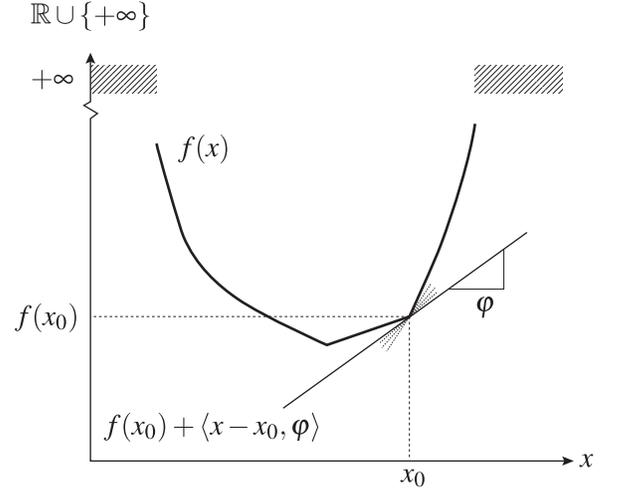}
  \caption{Illustration of the subdifferential for an 
    $f\in\Gamma_0(\RR)$. For a $x_0\in \RR$, $\partial f(x_0)$ is a
    collection of slopes of tangent functionals. One such slope
    $\varphi$ and its affine mapping is shown explicitly, the
    rest is indicated with stippled lines. $\varphi$ is not unique
    since the graph of $f$ has a ``kink'' at $x_0$.\label{fig:tangents}}
\end{figure}

In the context of DFT, $F$ is subdifferentiable at $\rho \in X$ if and only
if $\rho$ is the ground-state density of a potential $v \in X^\ast$, 
\begin{equation}
E(v) =
F(\rho) + (v|\rho) = \inf_{\rho'} \left( F(\rho') + (v|\rho') \right)
\end{equation}
so that
\begin{equation}
  F(\rho') \geq F(\rho) + (v|\rho'-\rho), \quad \forall \rho \in X.
\end{equation}
By the Hohenberg--Kohn theorem, we know that
\begin{equation}
  \partial F(\rho) = \left\{ v + \mu \; : \; \mu \in \RR\right\}
\end{equation}
if $\rho$ is $v$-representable and that $\partial F(\rho) =
\emptyset$ otherwise. Thus, from the point of view of convex analysis, the notion of
$v$-representability of $\rho$ is equivalent to subdifferentiability
of $F$ at $\rho$. It follows that the $v$-representable densities are dense in
the set of $N$-representable densities, the effective domain of $F$.

\subsection{G\^ateaux differentiability}

Let $x,y\in X$. The \emph{directional derivative} of $f$ at $x$ in the
direction of $y$ is defined by
\begin{equation}
  f'(x;y) := \lim_{\epsilon\rightarrow 0^+} \epsilon^{-1}[f(x +
  \epsilon y) - f(x)]
\end{equation}
if this limit exists ($+\infty$ is accepted as limit).
For $f\in\Gamma_0(X)$, 
the directional derivative $f'(x;y)$ always exists.

Let $x\in X$ be given. If there is a $\varphi \in X^*$ such that
\begin{equation}
  f'(x;y) = \braket{\varphi,y}, \quad \forall y \in X
\end{equation}
then $f$ is said to be \emph{G\^ateaux differentiable} at $x$. In other
words, a function is  G\^ateaux differentiable if its various directional derivatives may be assembled into a
linear functional at $x$. 
The G\^ateaux derivative is the usual notion of functional derivative
encountered in the calculus of variations, for which we write
$\varphi = \delta f(x)/\delta x$.

If $f$ is continuous and has a unique subgradient at $x$, then it is
is also G\^ateaux differentiable at $x$; the converse statement is
also true, but note that a unique subgradient alone is not enough to ensure 
G\^ateaux differentiability: continuity is not implied by a unique
subgradient.

\subsection{Fr\'echet differentiability}

A stronger notion of differentiability is given by the Fr\'echet
derivative. Let $x\in X$. If there exists $\varphi\in X^*$ such that
for all sequences $h_n\rightarrow 0$ in $X$ as $n \to \infty$, 
\begin{equation}
\lim_{n \to \infty} \frac{|f(x + h_n) - f(x) - \braket{\varphi,h_n}|}{\|h_n\|}  = 0,
\end{equation}
then $f$ is \emph{Fr\'echet differentiable} at $x$, and $\nabla f(x) = \varphi$ is
the Fr\'echet derivative.

Clearly, Fr\'echet differentiable implies 
G\^ateaux differentiable, but not the other way around. In fact, if
$\nabla f(x)$ exists at $x$, then
\begin{equation}
  f(x+h) = f(x) + \braket{\nabla f(x),h} + o(\|h\|),
\end{equation}
so that $f$ is approximated by its linearization around $x$. This is not
true if $f$ is merely G\^ateaux differentiable.

\bibliographystyle{aipnum4-1}

\begin{thebibliography}{19}%
\makeatletter
\providecommand \@ifxundefined [1]{%
 \@ifx{#1\undefined}
}%
\providecommand \@ifnum [1]{%
 \ifnum #1\expandafter \@firstoftwo
 \else \expandafter \@secondoftwo
 \fi
}%
\providecommand \@ifx [1]{%
 \ifx #1\expandafter \@firstoftwo
 \else \expandafter \@secondoftwo
 \fi
}%
\providecommand \natexlab [1]{#1}%
\providecommand \enquote  [1]{``#1''}%
\providecommand \bibnamefont  [1]{#1}%
\providecommand \bibfnamefont [1]{#1}%
\providecommand \citenamefont [1]{#1}%
\providecommand \href@noop [0]{\@secondoftwo}%
\providecommand \href [0]{\begingroup \@sanitize@url \@href}%
\providecommand \@href[1]{\@@startlink{#1}\@@href}%
\providecommand \@@href[1]{\endgroup#1\@@endlink}%
\providecommand \@sanitize@url [0]{\catcode `\\12\catcode `\$12\catcode
  `\&12\catcode `\#12\catcode `\^12\catcode `\_12\catcode `\%12\relax}%
\providecommand \@@startlink[1]{}%
\providecommand \@@endlink[0]{}%
\providecommand \url  [0]{\begingroup\@sanitize@url \@url }%
\providecommand \@url [1]{\endgroup\@href {#1}{\urlprefix }}%
\providecommand \urlprefix  [0]{URL }%
\providecommand \Eprint [0]{\href }%
\providecommand \doibase [0]{http://dx.doi.org/}%
\providecommand \selectlanguage [0]{\@gobble}%
\providecommand \bibinfo  [0]{\@secondoftwo}%
\providecommand \bibfield  [0]{\@secondoftwo}%
\providecommand \translation [1]{[#1]}%
\providecommand \BibitemOpen [0]{}%
\providecommand \bibitemStop [0]{}%
\providecommand \bibitemNoStop [0]{.\EOS\space}%
\providecommand \EOS [0]{\spacefactor3000\relax}%
\providecommand \BibitemShut  [1]{\csname bibitem#1\endcsname}%
\let\auto@bib@innerbib\@empty
\bibitem [{\citenamefont {Hohenberg}\ and\ \citenamefont
  {Kohn}(1964)}]{Hohenberg1964}%
  \BibitemOpen
  \bibfield  {author} {\bibinfo {author} {\bibfnamefont {P.}~\bibnamefont
  {Hohenberg}}\ and\ \bibinfo {author} {\bibfnamefont {W.}~\bibnamefont
  {Kohn}},\ }\href@noop {} {\bibfield  {journal} {\bibinfo  {journal} {Phys.
  Rev.}\ }\textbf {\bibinfo {volume} {136}},\ \bibinfo {pages} {B864} (\bibinfo
  {year} {1964})}\BibitemShut {NoStop}%
\bibitem [{\citenamefont {Lieb}(1983)}]{Lieb1983}%
  \BibitemOpen
  \bibfield  {author} {\bibinfo {author} {\bibfnamefont {E.~H.}\ \bibnamefont
  {Lieb}},\ }\href@noop {} {\bibfield  {journal} {\bibinfo  {journal} {Int. J.
  Quant. Chem.}\ }\textbf {\bibinfo {volume} {{24}}},\ \bibinfo {pages} {243}
  (\bibinfo {year} {1983})}\BibitemShut {NoStop}%
\bibitem [{\citenamefont {Schuch}\ and\ \citenamefont
  {Verstraete}(2009)}]{Schuch2009}%
  \BibitemOpen
  \bibfield  {author} {\bibinfo {author} {\bibfnamefont {N.}~\bibnamefont
  {Schuch}}\ and\ \bibinfo {author} {\bibfnamefont {F.}~\bibnamefont
  {Verstraete}},\ }\href@noop {} {\bibfield  {journal} {\bibinfo  {journal}
  {Nature Physics}\ }\textbf {\bibinfo {volume} {5}},\ \bibinfo {pages} {732}
  (\bibinfo {year} {2009})}\BibitemShut {NoStop}%
\bibitem [{\citenamefont {Garey}\ and\ \citenamefont
  {Johnson}(1979)}]{Garey1979}%
  \BibitemOpen
  \bibfield  {author} {\bibinfo {author} {\bibfnamefont {M.}~\bibnamefont
  {Garey}}\ and\ \bibinfo {author} {\bibfnamefont {D.}~\bibnamefont
  {Johnson}},\ }\href@noop {} {\emph {\bibinfo {title} {Computers and
  Intractability: A Guide to the Theory of NP-Completeness}}}\ (\bibinfo
  {publisher} {W.H. Freeman and Company},\ \bibinfo {year} {1979})\BibitemShut
  {NoStop}%
\bibitem [{\citenamefont {Lammert}(2005)}]{Lammert2005}%
  \BibitemOpen
  \bibfield  {author} {\bibinfo {author} {\bibfnamefont {P.}~\bibnamefont
  {Lammert}},\ }\href@noop {} {\bibfield  {journal} {\bibinfo  {journal} {Int.
  J. Quant. Chem.}\ }\textbf {\bibinfo {volume} {107}},\ \bibinfo {pages}
  {1944} (\bibinfo {year} {2005})}\BibitemShut {NoStop}%
\bibitem [{\citenamefont {Kohn}\ and\ \citenamefont
  {Sham}(1965)}]{KohnSham1965}%
  \BibitemOpen
  \bibfield  {author} {\bibinfo {author} {\bibfnamefont {W.}~\bibnamefont
  {Kohn}}\ and\ \bibinfo {author} {\bibfnamefont {L.~J.}\ \bibnamefont
  {Sham}},\ }\href@noop {} {\bibfield  {journal} {\bibinfo  {journal} {Phys.
  Rev.}\ }\textbf {\bibinfo {volume} {140}},\ \bibinfo {pages} {A1133}
  (\bibinfo {year} {1965})}\BibitemShut {NoStop}%
\bibitem [{\citenamefont {Bauschke}\ and\ \citenamefont
  {Combettes}(2011)}]{BauschkeAndCombettes}%
  \BibitemOpen
  \bibfield  {author} {\bibinfo {author} {\bibfnamefont {H.}~\bibnamefont
  {Bauschke}}\ and\ \bibinfo {author} {\bibfnamefont {P.}~\bibnamefont
  {Combettes}},\ }\href@noop {} {\emph {\bibinfo {title} {Convex Analysis and
  Monotone Operator Theory in Hilbert Spaces}}}\ (\bibinfo  {publisher}
  {Springer},\ \bibinfo {address} {New York, Dordrecht, Heidelberg, London},\
  \bibinfo {year} {2011})\BibitemShut {NoStop}%
\bibitem [{\citenamefont {Levy}(1979)}]{Levy1979}%
  \BibitemOpen
  \bibfield  {author} {\bibinfo {author} {\bibfnamefont {M.}~\bibnamefont
  {Levy}},\ }\href@noop {} {\bibfield  {journal} {\bibinfo  {journal} {Proc.
  Natl. Acad. Sci.}\ }\textbf {\bibinfo {volume} {76}},\ \bibinfo {pages}
  {6062} (\bibinfo {year} {1979})}\BibitemShut {NoStop}%
\bibitem [{\citenamefont {Evans}(1998)}]{Evans1998}%
  \BibitemOpen
  \bibfield  {author} {\bibinfo {author} {\bibfnamefont {L.}~\bibnamefont
  {Evans}},\ }\href@noop {} {\emph {\bibinfo {title} {Partial Differential
  Equations}}}\ (\bibinfo  {publisher} {American Mathematical Society},\
  \bibinfo {address} {Providence, R.I.},\ \bibinfo {year} {1998})\BibitemShut
  {NoStop}%
\bibitem [{\citenamefont {Babuska}\ and\ \citenamefont
  {Osborn}(1989)}]{Babuska1989}%
  \BibitemOpen
  \bibfield  {author} {\bibinfo {author} {\bibfnamefont {I.}~\bibnamefont
  {Babuska}}\ and\ \bibinfo {author} {\bibfnamefont {J.}~\bibnamefont
  {Osborn}},\ }\href@noop {} {\bibfield  {journal} {\bibinfo  {journal} {Math.
  Comp.}\ }\textbf {\bibinfo {volume} {52}},\ \bibinfo {pages} {275} (\bibinfo
  {year} {1989})}\BibitemShut {NoStop}%
\bibitem [{\citenamefont {Zhao}, \citenamefont {Morrison},\ and\ \citenamefont
  {Parr}(1994)}]{ZMP}%
  \BibitemOpen
  \bibfield  {author} {\bibinfo {author} {\bibfnamefont {Q.}~\bibnamefont
  {Zhao}}, \bibinfo {author} {\bibfnamefont {R.}~\bibnamefont {Morrison}}, \
  and\ \bibinfo {author} {\bibfnamefont {R.}~\bibnamefont {Parr}},\ }\href@noop
  {} {\bibfield  {journal} {\bibinfo  {journal} {Phys. Rev. A}\ }\textbf
  {\bibinfo {volume} {50}},\ \bibinfo {pages} {2138} (\bibinfo {year}
  {1994})}\BibitemShut {NoStop}%
\bibitem [{\citenamefont {Yang}\ and\ \citenamefont {Wu}(2002)}]{YangWu}%
  \BibitemOpen
  \bibfield  {author} {\bibinfo {author} {\bibfnamefont {W.}~\bibnamefont
  {Yang}}\ and\ \bibinfo {author} {\bibfnamefont {Q.}~\bibnamefont {Wu}},\
  }\href@noop {} {\bibfield  {journal} {\bibinfo  {journal} {Phys. Rev. Lett.}\
  }\textbf {\bibinfo {volume} {89}},\ \bibinfo {pages} {143002} (\bibinfo
  {year} {2002})}\BibitemShut {NoStop}%
\bibitem [{\citenamefont {Heaton-Burgess}, \citenamefont {Bulat},\ and\
  \citenamefont {Yang}(2007)}]{PRLOEPReg}%
  \BibitemOpen
  \bibfield  {author} {\bibinfo {author} {\bibfnamefont {T.}~\bibnamefont
  {Heaton-Burgess}}, \bibinfo {author} {\bibfnamefont {F.~A.}\ \bibnamefont
  {Bulat}}, \ and\ \bibinfo {author} {\bibfnamefont {W.}~\bibnamefont {Yang}},\
  }\href@noop {} {\bibfield  {journal} {\bibinfo  {journal} {Phys. Rev. Lett.}\
  }\textbf {\bibinfo {volume} {98}},\ \bibinfo {pages} {256401} (\bibinfo
  {year} {2007})}\BibitemShut {NoStop}%
\bibitem [{\citenamefont {Wu}\ and\ \citenamefont {Yang}(2003)}]{WuYang}%
  \BibitemOpen
  \bibfield  {author} {\bibinfo {author} {\bibfnamefont {Q.}~\bibnamefont
  {Wu}}\ and\ \bibinfo {author} {\bibfnamefont {W.}~\bibnamefont {Yang}},\
  }\href@noop {} {\bibfield  {journal} {\bibinfo  {journal} {J. Chem. Phys.}\
  }\textbf {\bibinfo {volume} {118}},\ \bibinfo {pages} {2498} (\bibinfo {year}
  {2003})}\BibitemShut {NoStop}%
\bibitem [{\citenamefont {Bulat}\ \emph {et~al.}(2007)\citenamefont {Bulat},
  \citenamefont {Heaton-Burgess}, \citenamefont {Cohen},\ and\ \citenamefont
  {Yang}}]{JCPBulat}%
  \BibitemOpen
  \bibfield  {author} {\bibinfo {author} {\bibfnamefont {F.~A.}\ \bibnamefont
  {Bulat}}, \bibinfo {author} {\bibfnamefont {T.}~\bibnamefont
  {Heaton-Burgess}}, \bibinfo {author} {\bibfnamefont {A.~J.}\ \bibnamefont
  {Cohen}}, \ and\ \bibinfo {author} {\bibfnamefont {W.}~\bibnamefont {Yang}},\
  }\href@noop {} {\bibfield  {journal} {\bibinfo  {journal} {J. Chem. Phys.}\
  }\textbf {\bibinfo {volume} {127}},\ \bibinfo {pages} {174101} (\bibinfo
  {year} {2007})}\BibitemShut {NoStop}%
\bibitem [{\citenamefont {Teale}, \citenamefont {Coriani},\ and\ \citenamefont
  {Helgaker}(2009)}]{US}%
  \BibitemOpen
  \bibfield  {author} {\bibinfo {author} {\bibfnamefont {A.~M.}\ \bibnamefont
  {Teale}}, \bibinfo {author} {\bibfnamefont {S.}~\bibnamefont {Coriani}}, \
  and\ \bibinfo {author} {\bibfnamefont {T.}~\bibnamefont {Helgaker}},\
  }\href@noop {} {\bibfield  {journal} {\bibinfo  {journal} {J. Chem. Phys.}\
  }\textbf {\bibinfo {volume} {130}},\ \bibinfo {pages} {104111} (\bibinfo
  {year} {2009})}\BibitemShut {NoStop}%
\bibitem [{\citenamefont {van Tiel}(1984)}]{VanTiel}%
  \BibitemOpen
  \bibfield  {author} {\bibinfo {author} {\bibfnamefont {J.}~\bibnamefont {van
  Tiel}},\ }\href@noop {} {\emph {\bibinfo {title} {Convex Analysis, an
  Introductory Text}}}\ (\bibinfo  {publisher} {Wiley},\ \bibinfo {address}
  {Chichester},\ \bibinfo {year} {1984})\BibitemShut {NoStop}%
\bibitem [{\citenamefont {Ekeland}\ and\ \citenamefont
  {T\'emam}(1999)}]{EkelandAndTemam}%
  \BibitemOpen
  \bibfield  {author} {\bibinfo {author} {\bibfnamefont {I.}~\bibnamefont
  {Ekeland}}\ and\ \bibinfo {author} {\bibfnamefont {R.}~\bibnamefont
  {T\'emam}},\ }\href@noop {} {\emph {\bibinfo {title} {Convex Analysis and
  Variational Problems}}}\ (\bibinfo  {publisher} {SIAM},\ \bibinfo {address}
  {Philadelphia},\ \bibinfo {year} {1999})\BibitemShut {NoStop}%
\bibitem [{\citenamefont {Kreyszig}(1989)}]{Kreyszig}%
  \BibitemOpen
  \bibfield  {author} {\bibinfo {author} {\bibfnamefont {E.}~\bibnamefont
  {Kreyszig}},\ }\href@noop {} {\emph {\bibinfo {title} {Introductory
  Functional Analysis with Applications}}}\ (\bibinfo  {publisher} {Wiley},\
  \bibinfo {address} {Chichester},\ \bibinfo {year} {1989})\BibitemShut
  {NoStop}%
\end{thebibliography}

%

\end{document}